\documentclass[lettersize,journal]{IEEEtran}
\usepackage{amsmath,amsfonts}
\usepackage{algorithmic}
\usepackage{algorithm}
\usepackage{array}
\usepackage[caption=false,font=normalsize,labelfont=sf,textfont=sf]{subfig}
\usepackage{textcomp}
\usepackage{stfloats}
\usepackage{url}
\usepackage{verbatim}
\usepackage{graphicx}
\usepackage{cite}
\usepackage{amsthm}
\usepackage{cite}
\usepackage{amsmath,amssymb,amsfonts}
\usepackage{graphicx}
\usepackage{algorithm,algorithmic}
\usepackage{hyperref}
\hypersetup{hidelinks=true}
\usepackage{textcomp}

\usepackage{color}
\usepackage{times}
\usepackage{epsfig}
\usepackage{mathrsfs}
\usepackage{epstopdf}
\usepackage{subfig}
\usepackage{float}

\newtheorem{thm}{Theorem}

\newtheorem{rem}{Remark}

\newtheorem{lem}{Lemma}
\newtheorem{ass}{Assumption}

\newcommand{\bea}{\begin{eqnarray}}
\newcommand{\eea}{\end{eqnarray}}

\newcommand{\beaN}{\begin{eqnarray*}}
\newcommand{\eeaN}{\end{eqnarray*}}

\def\diff{\mathrm{d}}
\def\Z{\mathbb{Z}}
\def\R{\mathbb{R}}

\begin{document}

\title{Segment-Based Two-Loop Adaptive Iterative Learning Control for Spacecraft Position and Attitude Tracking}

\author{Fan Zhang, Deyuan Meng, \IEEEmembership{Senior Member, IEEE}, and Ying Tan, \IEEEmembership{Fellow, IEEE}
\thanks{Fan Zhang and Deyuan Meng are with the Seventh Research Division, Beihang University (BUAA), Beijing 100191, China, and also with the School of Automation Science and Electrical Engineering, Beihang University (BUAA), Beijing 100191, China (email: zhangfan\_nwpu@nwpu.edu.cn, dymeng@buaa.edu.cn).}
\thanks{Ying Tan is with the Department of Mechanical Engineering,
University of Melbourne, Parkville, VIC 3010, Australia (e-mail:
yingt@unimelb.edu.au).}}

\maketitle

\begin{abstract}
Proximity operations of rigid bodies, such as spacecraft rendezvous and docking, require precise tracking of both position and attitude over finite time intervals. These operations are often repeated under uncertain conditions, with unknown but repeatable parameters and disturbances. Adaptive iterative learning control (ILC) is well suited to such tasks, as it can track desired trajectories while learning unknown, iteration-invariant signals or parameters. However, conventional adaptive ILC faces two challenges: (i) the coupling between rotational and translational dynamics complicates the design of the two coordinated learning loops for position and attitude, and (ii) standard adaptive ILC designs cannot guarantee bounded control inputs. To address these issues, we propose a dual-number-based, segment-based two-loop adaptive ILC framework for simultaneous high-precision position and attitude tracking. The framework employs two learning loops that interact through a dual-number representation of tracking errors, combining position and attitude errors into a single mathematical object for unified control design. A segment-based dynamic projection mechanism ensures that both parameter estimates and control inputs remain bounded without prior knowledge of uncertainties. Mathematical analysis and numerical simulations demonstrate that the proposed framework significantly enhances tracking performance under unknown but repeatable uncertainties and strong rotational-translational coupling. 
\end{abstract}

\begin{IEEEkeywords}
Adaptive iterative learning control (ILC), two-loop ILC, dual-number representation, segment-based projection spacecraft proximity operations.
\end{IEEEkeywords}

\section{Introduction}\label{sec_Introduction}

Proximity operations involve controlling one rigid body relative to another at distances where small errors in position or attitude can compromise safety or mission success. In aerospace, this includes spacecraft rendezvous and docking, aircraft aerial refueling, and formation flying \cite{cszflsz:25,dqrxc:18,zzc:21}. These tasks require precise control under uncertain and varying dynamic conditions. High-precision tracking is particularly critical in missions such as NASA's CubeSat Proximity Operations Demonstration (CPOD), which used CubeSats---small, modular satellites roughly $10 \times 10 \times 10$ cm per unit---for autonomous rendezvous and docking demonstrations \cite{rwm:18}.

Maintaining this precision is challenging, yet many operations are repeated under similar conditions: although the environment may vary, the underlying task dynamics remain largely consistent. This repetition forms the basis for iterative learning control (ILC), which leverages past executions to progressively improve tracking and reduce errors. Standard ILC systematically incorporates previous iteration errors into the control input, effectively creating a feedforward learning mechanism that improves transient performance while reducing the dependence on complete system knowledge \cite{bta:06}. When partial model knowledge is available but some parameters or iteration-invariant disturbances are unknown, adaptive ILC estimates these unknowns online, allowing faster convergence, improved tracking, and better disturbance rejection \cite{t:04,lsx:19,xt:02}.

Despite its promise, conventional adaptive ILC faces two major challenges for complex, multi-degree-of-freedom tasks such as rigid body proximity operations: (i) strong dynamic coupling between position and attitude complicates the design of the two coordinated learning loops, as actions in one loop affect the other and impede simultaneous convergence, and (ii) standard adaptive ILC cannot inherently guarantee bounded control inputs and parameter estimates, leaving the amplitude of learning updates potentially unbounded \cite{Dixon:03}.

To address these challenges, this paper develops a segment-wise two-loop adaptive ILC framework with dual numbers. The objectives are to (i) achieve simultaneous, high-precision position and attitude tracking and (ii) ensure uniform boundedness of both control inputs and parameter estimates without prior knowledge of uncertainties. Our contributions are as follows.

\begin{enumerate}
    \item Coordinated two-loop adaptive ILC: Position and attitude loops are coordinated via a dual-number representation of tracking errors, encoding position in the real part and attitude in the dual part. This design respects rigid body geometry and maintains consistency under rotational-translational coupling.
    \item Segment-based dynamic projection: The time horizon is divided into finite segments, dynamically constraining parameter estimates and control inputs to ensure $\mathscr{L}_{\infty\mathrm{e}}$-norm boundedness (i.e., $\sup_{0\leqslant t\leqslant T} \|\boldsymbol{x}(t)\|<\infty$ for some $\boldsymbol{x}(t)\in\R^n$) without prior uncertainty bounds.
    \item Enhanced tracking with guaranteed boundedness: Mathematical analysis and numerical validation demonstrate high-precision tracking and boundedness under unknown parameters, repetitive disturbances, and strong coupling, effectively addressing the limitations of conventional adaptive ILC.
\end{enumerate}

The remainder of the paper is structured as follows. Section \ref{sec_Preliminaries} introduces quaternions and dual quaternions for coupled position-attitude representation. Section \ref{sec_dynamics} presents the kinematics and dynamics, based on which we formulate the tracking problem. Section \ref{sec_controller} presents the two-loop adaptive ILC design, dual-number representation, segment-based projection, and boundedness analysis. The simulation results are in Section \ref{sec_Sim}, and the conclusions in Section \ref{sec_Conclusion}.

{\it Notations:} Let $\mathbb{R}$ denote the set of real numbers, $\mathbb{Z}$ the set of integers, $\mathbb{Z}_{+}=\{0,1,2,\ldots\}$ the set of nonnegative integers, $\mathbb{Z}_{++}=\{1,2,\ldots\}$ the set of positive integers, and $\mathbb{Z}_{m}=\{1,2,\ldots,m\}$ the set of integers from $1$ to $m$ for any $m \geqslant 1$. For any matrix $\boldsymbol{A} = [a_{ij}] \in \mathbb{R}^{m \times n}$, the spectral norm of $\boldsymbol{A}$ is denoted as $\|\boldsymbol{A}\|_2$ and is equal to $\sigma_1(\boldsymbol{A})$, which represents the largest singular value of $\boldsymbol{A}$. The $n$-sphere is defined as $\mathbb{S}^{n}=\left\{\boldsymbol{x} \in \mathbb{R}^{n+1} \mid \|\boldsymbol{x}\|_2=1 \right\}$. The symbol $\boldsymbol{I}_{3} \in \mathbb{R}^{3 \times 3}$ denotes the $3 \times 3$ identity matrix, $\boldsymbol{1}$ denotes the column vector $[1,0,0,0]^\top$, and $\boldsymbol{0}$ represents the zero vector in $\mathbb{R}^{4}$. For any vector $\boldsymbol{x} = \left[x_1,x_2,x_3\right]^\top \in \mathbb{R}^3$, the cross-product operator and the dimension-augmentation operator are defined as, respectively,
\begin{align}
\label{eq1_f}
\boldsymbol{x}^{\times}=\left[\begin{array}{ccc}
0 & -x_3 & x_2 \\
x_3 & 0 & -x_1 \\
-x_2 & x_1 & 0
\end{array}\right]
\text{ and }
\operatorname{aug}\left(\boldsymbol{x}\right)\triangleq\left[\begin{array}{c}
0  \\
\boldsymbol{x}
\end{array}\right].
\end{align}
\noindent For any four-dimensional vector $\boldsymbol{y} = \left[y_1,y_2,y_3,y_4\right]^\top \in \mathbb{R}^4$, its corresponding dimension-reduction operator to the specified three-dimensional space is defined as
\begin{align}
\label{eq2_f}
\operatorname{red}\left(\boldsymbol{y}\right)\triangleq[y_2,\,y_3,\,y_4]^\top.
\end{align}
For any $z \in \mathbb{R}$, the sign function $\operatorname{sgn}(z)$ is defined as
\begin{equation}
\label{eq3_f}
\operatorname{sgn}(z) =
\begin{cases}
-1,\,\text{if } z < 0, \\
0,\,\text{if } z = 0, \\
1,\,\text{if } z > 0 .
\end{cases}
\end{equation}
For any vector $\boldsymbol{z} = [z_i] \in \mathbb{R}^n$, we extend the definition component-wise as
\begin{equation}
\label{eq4_f}
\operatorname{sgn}(\boldsymbol{z}) \triangleq [ \operatorname{sgn}(z_i)] \in \mathbb{R}^n.
\end{equation}

\section{Quaternions and Dual Quaternions}\label{sec_Preliminaries}

\par This section sets the foundations by introducing the fundamental concepts and mathematical preliminaries necessary for understanding the rotation and translation of rigid bodies. We began by defining two general orthogonal coordinate frames denoted as $\mathcal{F}_1$ and $\mathcal{F}_2$. Based on these definitions, we utilize \textit{quaternions} and \textit{dual quaternions} for representing the position and attitude of $\mathcal{F}_2$ with respect to $\mathcal{F}_1$ of rigid bodies. Additionally, this section presents several lemmas pertinent to dual quaternions, which furnish the mathematical framework necessary for the analysis and control of coupled motion in fields such as aerospace engineering and robotics.

Throughout this paper, we emphasize the importance of clearly specifying both the origins and the frames in which vectors are expressed, as this distinction is crucial for avoiding ambiguity in rigid-body motion analysis.

\subsection{Quaternions and Unit Quaternions}

In this subsection, we formally define quaternions and unit quaternions, which provide the algebraic foundation, including the conjugation and multiplication, for representing three-dimensional rotations. Establishing these definitions is essential for constructing the motion kinematics and dynamics.

\par A {quaternion}, denoted by
\begin{align}
\label{eq5_f}
\boldsymbol{Q}\triangleq[\varepsilon,\boldsymbol{q}^{\top}]^{\top} \in \mathbb{R}^{4},
\end{align}
is composed of a scalar $\varepsilon\in\R$ and a vector $\boldsymbol{q}\in\R^3$. The corresponding conjugation of $\boldsymbol{Q}$ is defined as $\boldsymbol{Q}^{*}\triangleq[\varepsilon, -\boldsymbol{q}^{\top}]^{\top} \in \mathbb{R}^{4}$. The multiplication of two quaternions $\boldsymbol{Q}_i,\,i=1,2$ is given by
\begin{align}
\label{eq6_f}
\boldsymbol{Q}_1 \circ \boldsymbol{Q}_2\triangleq\left[\begin{array}{c}
\varepsilon_1 \varepsilon_2-\boldsymbol{q}_1^{\top} \boldsymbol{q}_2 \\
\varepsilon_1 \boldsymbol{q}_2+\varepsilon_2 \boldsymbol{q}_1+\boldsymbol{q}_1^{\times}\boldsymbol{q}_2
\end{array}\right],
\end{align}
\noindent which is {associative
and distributive but not commutative}. In addition, it is verified that
\begin{align}
\label{eq7_f}
\left(\boldsymbol{Q}_1 \circ \boldsymbol{Q}_2\right)^*=\boldsymbol{Q}_2^* \circ \boldsymbol{Q}_1^*.
\end{align}

\par According to Euler's theorem, formulated by Leonhard Euler in 1775, in three-dimensional space, any displacement of a rigid body such that a point on the rigid body remains fixed, is equivalent to a single rotation about some axis that runs through the fixed point. This means that the relative attitude of $\mathcal{F}_2$ with respect to $\mathcal{F}_1$ can be represented by a single rotation angle $\varphi \in \mathbb{R}$ (commonly measured in radians) about a fixed unit axis $\boldsymbol{e} \in \mathbb{R}^3$ \cite{Ruiter}. 
The positive direction of this rotation adheres to the right-hand rule.

The rotation angle and the unit axis can be converted to quaternion components as
\begin{align}
\label{eq7_f_add1}
\varepsilon \triangleq \cos\left(\frac{\varphi}{2} \right)~\text{and}~ 
\boldsymbol{q} \triangleq \boldsymbol{e} \sin \left(\frac{\varphi}{2}\right),
\end{align}
which allow a succinct formulation of the attitude kinematics and dynamics. 
The unit-norm constraint
\[
\varepsilon^2 + \boldsymbol{q}^\top \boldsymbol{q} = 1
\]
ensures that
\begin{align}
\label{eq8_f_add1}
\boldsymbol{Q} \circ \boldsymbol{Q}^* = \boldsymbol{1},
\end{align}
where $\boldsymbol{1}$ is defined in Notations of Section \ref{sec_Introduction}. The quaternion fulfilling  \eqref{eq8_f_add1} is referred to as the \emph{unit quaternion} \cite{h:04}, denoted by $\boldsymbol{Q} \in \mathbb{S}^3$. Unit quaternions are particularly useful for representing the orientation of rigid bodies, as both the rotation angle $\varphi$ and the unit axis $\boldsymbol{e}$ can be directly derived from $\boldsymbol{Q}$.

\subsection{Dual Quaternions and Unit Dual Quaternions}
\label{s2_ssb}

In this subsection, we introduce the concepts of dual quaternions and unit dual quaternions, with a particular emphasis on the fundamental role played by dual numbers. 

\par  {Dual quaternions} are generally understood as a combination of dual numbers and quaternions. The dual number is known as the complex number of parabolic type and extends the real numbers by adding a special imaginary unit $\epsilon$, which satisfies
\begin{align}
\label{eq8_f}
\epsilon^2=0 ~ \text{and} ~ \epsilon\neq 0.
\end{align}
\noindent The dual numbers form a commutative algebra of dimension two over the reals, which can be found in Section 2.4 of \cite{g:13}. The concept of dual vectors and dual matrices is based on the use of dual numbers. In this context, the set of dual vectors is denoted by $\mathbb{DR}^n=\{
\mathring{\boldsymbol{x}}
\triangleq \boldsymbol{x}_r+\epsilon \boldsymbol{x}_c:\boldsymbol{x}_r,\,\boldsymbol{x}_c\in\mathbb{R}^n\}$, where $\boldsymbol{x}_r,\,\boldsymbol{x}_c\in\R^n$ are called the real and imaginary (dual) parts, respectively. The definition of dual matrices follows a similar format. Given any $\mathring{\boldsymbol{x}}\triangleq \boldsymbol{x}_r+\epsilon \boldsymbol{x}_c,\,\mathring{\boldsymbol{y}}\triangleq \boldsymbol{y}_r+\epsilon \boldsymbol{y}_c\in\mathbb{DR}^n$, four distinct operators are defined as 
\begin{align}
\label{eq9_f}
\operatorname{crs}\left(\mathring{\boldsymbol{x}},\mathring{\boldsymbol{y}}\right)&\triangleq\boldsymbol{x}_r^\top\boldsymbol{y}_c+\boldsymbol{x}_c^\top\boldsymbol{y}_r
&
\operatorname{real}\left(\mathring{\boldsymbol{x}}\right)&\triangleq \boldsymbol{x}_r
\nonumber
\\
\operatorname{exch}\left(\mathring{\boldsymbol{x}}\right)&\triangleq \boldsymbol{x}_c+\epsilon \boldsymbol{x}_r
&
\operatorname{comp}\left(\mathring{\boldsymbol{x}}\right)&\triangleq \boldsymbol{x}_c,
\end{align}
\noindent and $\operatorname{sgn}\left(\cdot\right)$ can be extended to $\operatorname{sgn}\left(\mathring{\boldsymbol{x}}\right)\triangleq
\operatorname{sgn}\left(\boldsymbol{x}_r\right)+\epsilon \operatorname{sgn}\left(\boldsymbol{x}_c\right)$. Similarly, $\operatorname{red}\left(\cdot\right)$ can also be extended to $\operatorname{red}\left(\mathring{\boldsymbol{z}}\right)\triangleq
\operatorname{red}\left(\boldsymbol{z}_r\right)+\epsilon \operatorname{red}\left(\boldsymbol{z}_c\right)$ for all $\mathring{\boldsymbol{z}}\triangleq \boldsymbol{z}_r+\epsilon \boldsymbol{z}_c \in \mathbb{DR}^4$.

A dual quaternion is a special kind of dual vector and can be represented by
\begin{align}
\label{eq10_f}
\mathring{\boldsymbol{Q}} \triangleq \boldsymbol{Q}_r+\epsilon \boldsymbol{Q}_c \in \mathbb{DR}^4,
\end{align}
where $\boldsymbol{Q}_r,\,\boldsymbol{Q}_c\in\R^4$ are two ordinary quaternions defined in \eqref{eq5_f}. Based on the definition of quaternions in the previous subsection, it is possible to rewrite $\mathring{\boldsymbol{Q}}$ in \eqref{eq10_f} as
\begin{align}
\label{eq11_f}
\mathring{\boldsymbol{Q}}=[\mathring\varepsilon,\mathring{\boldsymbol{q}}^{\top}]^{\top},
\end{align}
where $\mathring\varepsilon\triangleq\varepsilon_r+\epsilon \varepsilon_c$ and $\mathring{\boldsymbol{q}}\triangleq\boldsymbol{q}_r+\epsilon \boldsymbol{q}_c$. Observing the similarity of \eqref{eq5_f} and \eqref{eq11_f}, we can naturally define the conjugation of $\mathring{\boldsymbol{Q}}$ as $\mathring{\boldsymbol{Q}}^{*}\triangleq[\mathring\varepsilon, -\mathring{\boldsymbol{q}}^{\top}]^{\top}\in \mathbb{DR}^{4}$, which can be rewritten as
\begin{align}
\label{eq12_f}
\mathring{\boldsymbol{Q}}^{*}= \boldsymbol{Q}^*_r+\epsilon \boldsymbol{Q}^*_c.
\end{align}
\noindent In addition, the multiplication of two dual quaternions $\mathring{\boldsymbol{Q}}_i,\,i=1,2$ follows the same rule presented in \eqref{eq6_f}, that is,
\begin{align}
\label{eq13_f}
\mathring{\boldsymbol{Q}}_1 \circ \mathring{\boldsymbol{Q}}_2\triangleq\left[\begin{array}{c}
\mathring\varepsilon_1 \mathring\varepsilon_2-\mathring{\boldsymbol{q}}_1^{\top} \mathring{\boldsymbol{q}}_2 \\
\mathring\varepsilon_1 \mathring{\boldsymbol{q}}_2+\mathring\varepsilon_2 \mathring{\boldsymbol{q}}_1+\mathring{\boldsymbol{q}}_1^{\times}\mathring{\boldsymbol{q}}_2
\end{array}\right],
\end{align}
\noindent and it is verified that
\begin{align}
\label{eq14_f}
\left(\mathring{\boldsymbol{Q}}_1 \circ \mathring{\boldsymbol{Q}}_2\right)^*=\mathring{\boldsymbol{Q}}_2^* \circ \mathring{\boldsymbol{Q}}_1^*.
\end{align}
\noindent Moreover, a dual quaternion $\mathring{\boldsymbol{Q}}$ is called a {unit dual quaternion} if $\mathring{\boldsymbol{Q}}\circ\mathring{\boldsymbol{Q}}^{*}=\boldsymbol{1}+\epsilon \boldsymbol{0}$ holds, and we denote the set of unit dual quaternions by $$\mathbb{DS}^3=\{\mathring{\boldsymbol{Q}}\in\mathbb{DR}^4:\mathring{\boldsymbol{Q}}\circ\mathring{\boldsymbol{Q}}^{*}=\boldsymbol{1}+\epsilon \boldsymbol{0}\}.$$
\noindent It is not difficult to verify that $\mathring{\boldsymbol{Q}}\circ\mathring{\boldsymbol{Q}}^{*}=\boldsymbol{1}+\epsilon \boldsymbol{0}$ is equivalent to $\mathring{\boldsymbol{Q}}^{*}\circ\mathring{\boldsymbol{Q}}=\boldsymbol{1}+\epsilon \boldsymbol{0}$ based on \eqref{eq13_f}. Additionally, we are able to demonstrate that, for any $\mathring{\boldsymbol{Q}}\in \mathbb{DS}^3$,
\begin{itemize}
\item the imaginary part of $\mathring{\boldsymbol{Q}}$ is orthogonal to its real part;
\item the $l_2$-norm of the real part of $\mathring{\boldsymbol{Q}}$ is equal to 1.
\end{itemize}

Based on both properties, we present the following lemma about unit dual quaternions.

\begin{lem}
\label{lem1}
Let $\mathring{\boldsymbol{Q}}_1$ and $\mathring{\boldsymbol{Q}}_2$ be unit dual quaternions, i.e., $\mathring{\boldsymbol{Q}}_1, \mathring{\boldsymbol{Q}}_2 \in \mathbb{DS}^3$. Then,
\[
\mathring{\boldsymbol{Q}}_1^* \circ \mathring{\boldsymbol{Q}}_2 \in \mathbb{DS}^3.
\] 
Moreover,
\[
\mathring{\boldsymbol{Q}}_1^* \circ \mathring{\boldsymbol{Q}}_2 = \boldsymbol{1} + \epsilon \boldsymbol{0} \quad \text{if and only if} \quad \mathring{\boldsymbol{Q}}_1 = \mathring{\boldsymbol{Q}}_2,
\] 
and
\[
\mathring{\boldsymbol{Q}}_1^* \circ \mathring{\boldsymbol{Q}}_2 = -\boldsymbol{1} + \epsilon \boldsymbol{0} \quad \text{if and only if} \quad \mathring{\boldsymbol{Q}}_1 = -\mathring{\boldsymbol{Q}}_2.
\]
\end{lem}

\begin{proof}
The proof is provided in Appendix \ref{appendix:lem1}.
\end{proof}

Operations such as conjugation and multiplication are known to preserve the unit dual quaternion structure. This ensures that the `difference' between two unit dual quaternions, measured via the product of one conjugate with the other, is itself a unit dual quaternion. This fact shows that unit dual quaternions are able to provide a unified representation of both rotation and translation.

\begin{rem}
This lemma formalizes how unit dual quaternions behave under conjugate multiplication. 
The product $\mathring{\boldsymbol{Q}}_1^* \circ \mathring{\boldsymbol{Q}}_2$ can be interpreted as the relative transformation from $\mathring{\boldsymbol{Q}}_1^*$ to $\mathring{\boldsymbol{Q}}_2$. 
If this product equals the identity, the two dual quaternions represent the same rigid-body pose; if it equals the negative identity, they represent the same pose but with opposite sign, which is physically equivalent. 
This property guarantees that control laws defined via dual quaternion differences remain geometrically consistent.
\end{rem}


\par As \emph{Chasles' theorem} states, the general displacement of a rigid body in space consists of a rotation about an axis (called the screw axis) and a translation parallel to that axis \cite{Chasles}. Based on this fact and on the definition of the Pl\"{u}cker line \cite{Pluker}, the position-attitude of frame $\mathcal{F}_2$ with respect to frame $\mathcal{F}_1$ can be compactly represented by a dual quaternion as
\begin{align}
\label{eq15_f}
\mathring{\boldsymbol{Q}} \triangleq \boldsymbol{Q} + \epsilon \frac{1}{2} \boldsymbol{Q} \circ \operatorname{aug}(\boldsymbol{P}),
\end{align}
where $\boldsymbol{Q} \in \mathbb{S}^3$ is the unit quaternion representing the rotation from $\mathcal{F}_1$ to $\mathcal{F}_2$, and $\boldsymbol{P} \in \mathbb{R}^3$ is the position vector pointing from the origin of $\mathcal{F}_1$ to the origin of $\mathcal{F}_2$, expressed in $\mathcal{F}_2$, while the operator $\operatorname{aug}(\cdot)$ is defined in \eqref{eq1_f} \cite{gws:21}.

The following lemma establishes the basic property of a dual quaternion used for representing rigid bodies' position and attitude.

\begin{lem}
\label{lem2}
Given a dual quaternion $\mathring{\boldsymbol{Q}}$ defined in \eqref{eq15_f}, we have $\mathring{\boldsymbol{Q}} \in \mathbb{DS}^3$, i.e., $\mathring{\boldsymbol{Q}}$ is a unit dual quaternion. Furthermore, the rotation $\boldsymbol{Q}$ and the translation $\boldsymbol{P}$ can be uniquely determined from $\mathring{\boldsymbol{Q}}$.
\end{lem}

\begin{proof}
The proof is provided in Appendix \ref{appendix:lem2}.
\end{proof}

\begin{rem}
\label{remark2}
According to Lemma \ref{lem2}, the dual quaternion $\mathring{\boldsymbol{Q}} \in \mathbb{DS}^3$ provides a consistent and compact representation of both rotation and translation for a rigid body. While alternative representations like $\boldsymbol{Q} + \epsilon \operatorname{aug}(\boldsymbol{P})$ might seem plausible, they introduce difficulties in subsequent control design.
\end{rem}

\begin{rem}
Lemma \ref{lem2} implies that dual quaternions rotate and translate vectors without altering their intrinsic properties. The scalar part of a transformed dual vector remains zero, consistent with the fact that physical position vectors do not have a scalar component. Meanwhile, the length of the vector is preserved, reflecting that rigid-body motions do not change distances in space. This property is essential for control applications, as it ensures that dual quaternion transformations maintain accurate and consistent tracking of both position and orientation.
\end{rem}

The following lemma establishes the calculation of the three-dimensional vector and the unit dual quaternions.

\begin{lem}
\label{lem4}
Let $\mathring{\boldsymbol{x}} \triangleq \boldsymbol{x}_r + \epsilon \boldsymbol{x}_c \in \mathbb{DR}^3$, $\mathring{\boldsymbol{y}} \triangleq \boldsymbol{y}_r + \epsilon \boldsymbol{y}_c \in \mathbb{DR}^4$, and $\mathring{\boldsymbol{Q}} \in \mathbb{DS}^3$. Suppose
\[
\mathring{\boldsymbol{y}} = \mathring{\boldsymbol{Q}}^* \circ \operatorname{aug}(\mathring{\boldsymbol{x}}) \circ \mathring{\boldsymbol{Q}}.
\]
Then the scalar part of $\mathring{\boldsymbol{y}}$ is $0 + \epsilon 0$, and $\|\boldsymbol{x}_r\|_2 = \|\boldsymbol{y}_r\|_2$.
\end{lem}

\begin{proof}
The proof is provided in Appendix \ref{appendix:lem4}.
\end{proof}

\begin{rem}
Lemma \ref{lem4} demonstrates that applying a unit dual quaternion transformation to a dual vector preserves essential geometric properties. Specifically, the scalar part remains zero, consistent with the physical interpretation of position vectors, and the length of the real vector is unchanged. This ensures that rigid-body motions represented by dual quaternions maintain distances and orientations correctly, which is crucial for accurate and physically consistent position-attitude tracking in control applications.
\end{rem}

\section{Motion Kinematics and Dynamics}
\label{sec_dynamics}

In this section, we first define three coordinate frames that facilitate the description of rigid-body motions in proximity operations. Fig. \ref{fig-1} provides an intuitive illustration of coupled position-attitude tracking, highlighting the interplay between translational and rotational motions. Building on the preliminaries introduced in Section \ref{sec_Preliminaries}, we then formulate the actual position-attitude kinematics and dynamics using a dual-number representation, which unifies position and attitude errors into a single mathematical object and enables coordinated analysis. This is followed by the error position-attitude kinematics and dynamics, which characterize the tracking process and form the basis for the dual-number-based adaptive ILC framework.

\begin{figure}[!hbtp]
	\centering
	\includegraphics[trim=9.5cm 1.3cm 9.5cm 1.3cm,width=0.35\textwidth,clip]{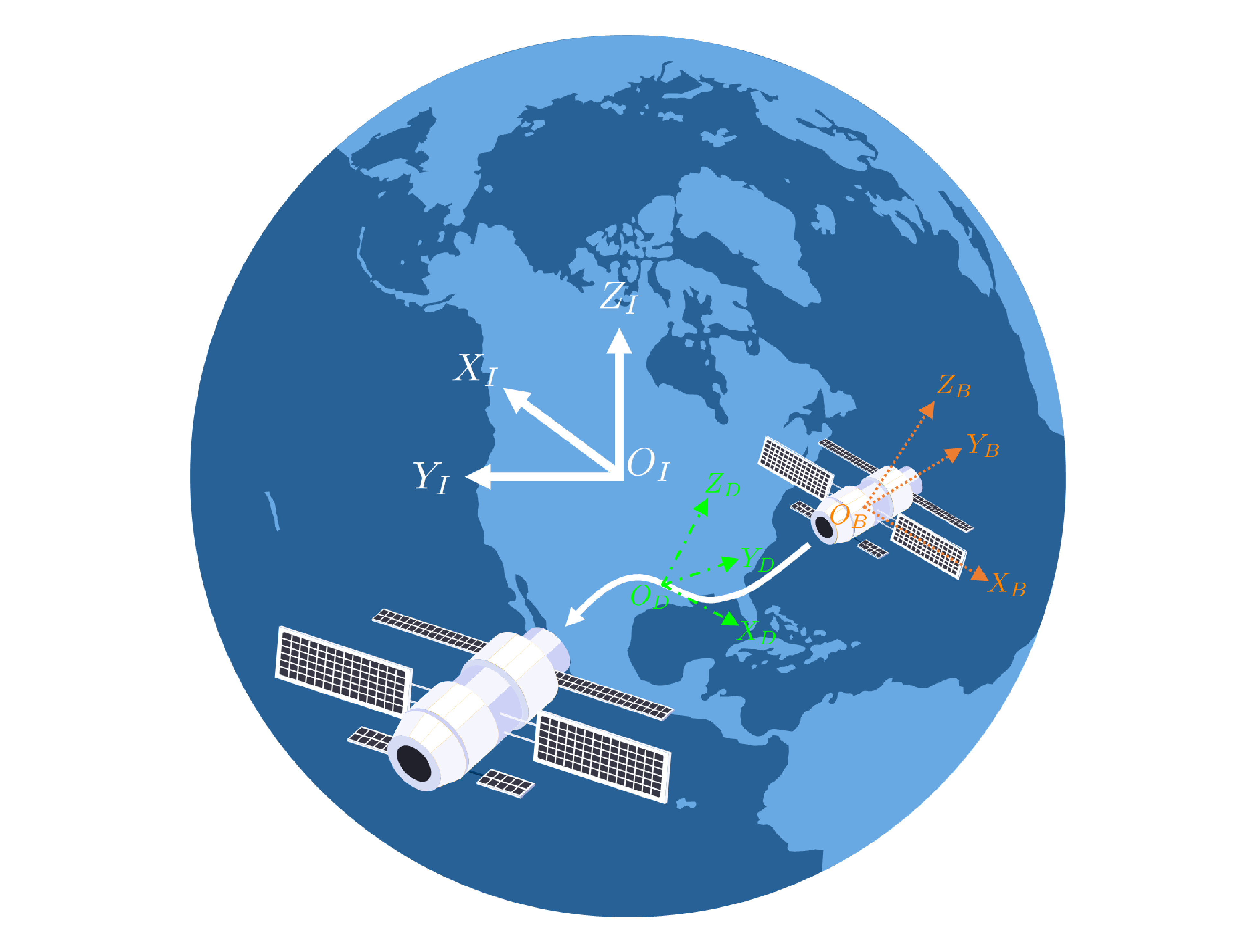}
	\caption{Demonstration of the couple position-attitude tracking in the rigid body proximity operation.
	}
	\label{fig-1}
\end{figure}
\subsection{Coordinate Frames}

To facilitate the subsequent design and analysis of position-attitude tracking, we use the definitions of three orthogonal coordinate frames as given in \cite{Ruiter}. 

\begin{enumerate}
    \item \textbf{Body-fixed frame $\mathcal{F}_B$ ($O_B-X_BY_BZ_B$):}
    \begin{itemize}
        \item \textit{Origin:} located at the center of mass of the rigid body.
        \item \textit{Orientation:} rigidly attached to the body.
        \item \textit{Property:} all points of the rigid body are fixed in this frame, consistent with the rigid body assumption.
    \end{itemize}

    \item \textbf{Inertial frame $\mathcal{F}_I$ ($O_I-X_IY_IZ_I$):}
    \begin{itemize}
        \item \textit{Origin:} located on Earth.
        \item \textit{Orientation:}fixed in space, not rotating with the Earth.
        \item \textit{Property:} any frame that is stationary or moves at constant velocity (without rotation) relative to an inertial frame is also considered an inertial frame.
    \end{itemize}
    \item \textbf{Desired reference frame $\mathcal{F}_D$ ($O_D-X_DY_DZ_D$):}
    \begin{itemize}
    \item \textit{Origin:} represents the desired position of the rigid body.
    \item \textit{Orientation:} represents the desired orientation of the rigid body.
    \end{itemize}
\end{enumerate}

These coordinate frames are illustrated in Fig. \ref{fig-1}. The green dash-dotted arrows represent the desired reference frame, the orange dotted arrows indicate the body-fixed frame, and the white solid arrows denote the inertial frame.

\subsection{Actual Attitude-Position Kinematics and Dynamics}

Based on the body-fixed frame $\mathcal{F}_B$ and the inertial frame $\mathcal{F}_I$, the motion of a rigid body is described using a unit dual quaternion, which compactly represents its orientation and position. The orientation is captured by a unit quaternion describing the rotation of the body-fixed frame relative to the inertial frame, while the position is represented by a vector from the inertial origin to the body-fixed origin, expressed in the body-fixed frame. By combining these into a single dual quaternion, both rotational and translational motions can be treated in a unified framework.

Based on \eqref{eq15_f}, at operation time $t \in [0, T]$ with terminal time $T<\infty$ during iteration $k \in \mathbb{Z}_{+}$, the rigid body's unit dual quaternion is given by
\begin{align}
\label{eq16_f}
\mathring{\boldsymbol{Q}}_k(t) \triangleq \boldsymbol{Q}_k(t) + \epsilon \frac{1}{2} \boldsymbol{Q}_k(t) \circ \operatorname{aug}(\boldsymbol{P}_k(t)),
\end{align}
where $\boldsymbol{Q}_k(t)\in\mathbb{S}^3$ denotes the quaternion representing orientation of $\mathcal{F}_B$ with respect to $\mathcal{F}_I$, $\boldsymbol{P}_k(t)\in\R^3$ denotes the position vector pointing from the origin of $\mathcal{F}_I$ to the origin of $\mathcal{F}_B$, expressed in $\mathcal{F}_B$. The reason for this construction has been explained in Remark \ref{remark2}.

Based on the definition of $\mathring{\boldsymbol{Q}}_k(t)$ in \eqref{eq16_f}, and consistent with the formulation presented in \cite{gws:21}, the position-attitude kinematics and dynamics are given by
\begin{align}
\dot{\mathring{\boldsymbol{Q}}}_k(t) &= \frac{1}{2} \mathring{\boldsymbol{Q}}_k(t) \circ \operatorname{aug}(\mathring{\boldsymbol{\omega}}_k(t)),
\label{eq17_f}
\\
\mathring{\boldsymbol{M}} \dot{\mathring{\boldsymbol{\omega}}}_k(t) &= -\mathring{\boldsymbol{\omega}}_k^\times(t) \mathring{\boldsymbol{M}} \mathring{\boldsymbol{\omega}}_k(t) + \mathring{\boldsymbol{f}}_k(t) + \mathring{\boldsymbol{d}}_k(t),
\label{eq21_f}
\end{align}
where the dual (twist) velocity is
\begin{align*}
\mathring{\boldsymbol{\omega}}_k(t) \triangleq \boldsymbol{\omega}_k(t) + \epsilon \boldsymbol{v}_k(t).
\end{align*}
In this expression, $\boldsymbol{\omega}_k(t)\in\R^3$ denotes the angular velocity of $\mathcal{F}_B$ with respect to $\mathcal{F}_I$, expressed in $\mathcal{F}_B$, and $\boldsymbol{v}_k(t)\in\R^3$ denotes the linear velocity of the origin of $\mathcal{F}_B$ with respect to the origin of $\mathcal{F}_I$, also expressed in $\mathcal{F}_B$, leading to 
\begin{align}
\label{eq20_f_add}
\boldsymbol{v}_k(t) = \dot{\boldsymbol{P}}_k(t) + \boldsymbol{\omega}_k^\times(t) \boldsymbol{P}_k(t).
\end{align}
The dual inertia matrix in \eqref{eq21_f} is defined as
\begin{align}
\label{eq22_f}
\mathring{\boldsymbol{M}} \triangleq m \boldsymbol{I}_3 \frac{\text{d}}{\text{d}\epsilon} + \epsilon \boldsymbol{J},
\end{align}
where $\frac{\text{d}}{\text{d}\epsilon}$ extracts the dual part of a dual vector and can be regarded as an operator. In \eqref{eq21_f}, $\mathring{\boldsymbol{f}}_k(t) \triangleq \boldsymbol{f}_k(t) + \epsilon \boldsymbol{\tau}_k(t)$ denotes the dual control input and $\mathring{\boldsymbol{d}}_k(t) \triangleq \boldsymbol{d}_k(t) + \epsilon \boldsymbol{s}_k(t)$ denotes the dual disturbance. In this context, $\boldsymbol{f}_k(t)\in\R^3$, $\boldsymbol{\tau}_k(t)\in\R^3$, $\boldsymbol{d}_k(t)\in\R^3$, and $ \boldsymbol{s}_k(t)\in\R^3$ represent the control force, control torque, disturbance force, and
disturbance torque expressed in $\mathcal{F}_B$, respectively.

\begin{rem}
Expanding the dynamics \eqref{eq21_f} explicitly shows the coupling between rotation and translation:
\begin{align*}
&{}m \ddot{\boldsymbol{P}}_k(t)
+
\epsilon
\boldsymbol{J}
\dot{\boldsymbol{\omega}}_k(t)
\nonumber
\\
={}&
-m
\left(
\dot{\boldsymbol{\omega}}_k^\times(t) \boldsymbol{P}_k(t)
+
2{\boldsymbol{\omega}}_k^\times(t) \dot{\boldsymbol{P}}_k(t)
+
\left({\boldsymbol{\omega}}_k^\times(t) \right)^2
 \boldsymbol{P}_k(t)
\right)
\nonumber
\\
&
+\boldsymbol{f}_k(t)+\boldsymbol{d}_k(t)
+
\epsilon
\left(
-
{\boldsymbol{\omega}}_k^\times(t)
\boldsymbol{J}
{\boldsymbol{\omega}}_k(t)
+\boldsymbol\tau_k(t)
+\boldsymbol{s}_k(t)
\right).
\end{align*}
The translational acceleration depends explicitly on the angular velocity and the angular acceleration, highlighting the intrinsic coupling between translation and rotation in the rigid-body motion.
\end{rem}

\begin{rem}
In \eqref{eq20_f_add}, it is important to note that $\boldsymbol{v}_k(t) = \dot{\boldsymbol{P}}_k(t)$ holds only if $\boldsymbol{\omega}_k(t)=\boldsymbol{0}$ and/or $\boldsymbol{P}_k(t)=\boldsymbol{0}$. This condition arises because $\boldsymbol{v}_k(t)$ and $\boldsymbol{P}_k(t)$ are not expressed in the inertial frame $\mathcal{F}_I$. On the other hand, if the spacecraft is stationary relative to the inertial frame, i.e., $\boldsymbol{v}(t)=\boldsymbol{0}$, it is not hard to verify that $\dot{\boldsymbol{P}}_k(t)
=-\boldsymbol{\omega}_k^\times(t) \boldsymbol{P}_k(t)$.
\end{rem}


\subsection{Error Attitude-Position Kinematics and Dynamics}

Following the discussion on the actual attitude-position kinematics and dynamics, this subsection proceeds to introduce the desired attitude-position kinematics and dynamics. Subsequently, we delve into the kinematic and dynamic equations that govern the errors in both attitude and position. The utilization of unit dual quaternions for a unified representation provides a coherent framework, significantly simplifying the subsequent control design processes.

Let $\mathring{\boldsymbol{Q}}_d(t) \in \mathbb{DS}^{3}$ denote the desired dual unit quaternion of $\mathcal{F}_D$ with respect to $\mathcal{F}_I$. Similarly, let $\mathring{\boldsymbol{\omega}}_d(t) \in \mathbb{DR}^{3}$ denote the desired dual velocity of $\mathcal{F}_D$ with respect to $\mathcal{F}_I$, expressed in $\mathcal{F}_D$. Therefore, it is evident that $\mathring{\boldsymbol{Q}}_d(t)$ and $\mathring{\boldsymbol{\omega}}_d(t)$ adhere to the same kinematics as presented in \eqref{eq17_f}, that is,
\begin{align}
\label{eq23_f}
\dot{\mathring{\boldsymbol{Q}}}_d(t)
=\frac{1}{2}\mathring{\boldsymbol{Q}}_d(t)\circ\operatorname{aug}\left(\mathring{\boldsymbol{\omega}}_d(t)\right).
\end{align}
\noindent Generally, the real and imaginary parts of $\mathring{\boldsymbol{{\omega}}}_d(t)$ and $\dot{\mathring{\boldsymbol{\omega}}}_d(t)$ are uniformly bounded for all $t \in [0, T]$. 

\par The dual unit quaternion error $\delta\mathring{\boldsymbol{Q}}_k(t)
\in \mathbb{DS}^3$ denotes the error information between $\mathring{\boldsymbol{Q}}_d(t)$ and $\mathring{\boldsymbol{Q}}_k(t)$. Specifically, based on \eqref{eq13_f} and Lemma \ref{lem1}, we have
\begin{align*}
\delta\mathring{\boldsymbol{Q}}_k(t)
\triangleq \delta{\boldsymbol{Q}}_k(t)
+
\epsilon \frac{1}{2} \delta{\boldsymbol{Q}}_k(t) \circ \operatorname{aug}\left(\delta{\boldsymbol{P}}_k(t)\right)
,
\end{align*}
where
\begin{align*}
\delta{\boldsymbol{Q}}_k(t) \triangleq {\boldsymbol{Q}}^{*}_d(t)
\circ
{\boldsymbol{Q}}_k(t)
\end{align*}
and
\begin{align*}
\delta{\boldsymbol{P}}_k(t) \triangleq {\boldsymbol{P}}_k(t) -
\operatorname{red}\left(
\delta{\boldsymbol{Q}}^*_k(t) \circ \operatorname{aug}\left({\boldsymbol{P}}_d(t)\right) \circ \delta{\boldsymbol{Q}}_k(t)\right).
\end{align*}
Based on this definition, it follows that
\begin{align}
\label{eq24_f}
\delta\mathring{\boldsymbol{Q}}_k(t)
=
\mathring{\boldsymbol{Q}}^{*}_d(t)
\circ
\mathring{\boldsymbol{Q}}_k(t) \in \mathbb{DS}^3,
\end{align}
\noindent which, in conjunction with \eqref{eq12_f}, \eqref{eq17_f}, and \eqref{eq23_f}, gives the derivative of $\delta\mathring{\boldsymbol{Q}}_k(t)$ as
\begin{align}
\label{eq25_f}
\delta\dot{\mathring{\boldsymbol{Q}}}_k(t)
=
\frac{1}{2}
\delta\mathring{\boldsymbol{Q}}_k(t)
\circ
\operatorname{aug}\left(
\delta\mathring{\boldsymbol{\omega}}_k(t)\right),
\end{align}
\noindent where the associative law of \eqref{eq13_f} and
\begin{align}
\label{eq26_f}
\operatorname{tr}
\left(\mathring{\boldsymbol{x}}
\right)
&\triangleq 
\operatorname{red}
\left(\delta\mathring{\boldsymbol{Q}}^{*}_k(t)
\circ
\operatorname{aug}\left(\mathring{\boldsymbol{x}}\right)
\circ
\delta\mathring{\boldsymbol{Q}}_k(t)\right),\,\forall \mathring{\boldsymbol{x}}\in\mathbb{DR}^3
\nonumber
\\
\delta\mathring{\boldsymbol{\omega}}_k(t)
&\triangleq \mathring{\boldsymbol{\omega}}_k(t)
-
\operatorname{tr}
\left(\mathring{\boldsymbol{\omega}}_{d}(t)
\right)
\end{align}
are utilized. Additionally, denoting the derivative of $\delta\mathring{\boldsymbol{\omega}}_k(t)$ as $\delta\dot{\mathring{\boldsymbol{\omega}}}_k(t)$, we obtain
\begin{align}
\label{eq27_f}
\mathring{\boldsymbol{M}}\delta\dot{\mathring{\boldsymbol{\omega}}}_k(t)
={}&
-\mathring{\boldsymbol{M}}
\left(\mathring{\boldsymbol{\omega}}^\times_k(t)
\delta{\mathring{\boldsymbol{\omega}}}_k(t)
\right)-\mathring{\boldsymbol{\omega}}^\times_k(t) \mathring{\boldsymbol{M}}\mathring{\boldsymbol{\omega}}_k(t)
\nonumber
\\
{}&
-\mathring{\boldsymbol{M}}
\operatorname{tr}\left(\dot{\mathring{\boldsymbol{\omega}}}_{d}(t)\right)
+\mathring{\boldsymbol{f}}_k(t)+\mathring{\boldsymbol{d}}_k(t)
,
\end{align}
\noindent where $\mathring{\boldsymbol{M}}$ is defined in \eqref{eq22_f}, and Lemma \ref{lem4} is employed. The derivation of \eqref{eq27_f} is detailed in Appendix \ref{appendix:eq25}. By leveraging \eqref{eq25_f} and \eqref{eq27_f}, the error position-attitude kinematics and dynamics are presented.

\begin{rem}
\label{rem2_add}
It is crucial to note that the distinction between $\mathring{\boldsymbol{Q}}_d(t)$ and $\mathring{\boldsymbol{Q}}_k(t)$ can be articulated through $\delta\mathring{\boldsymbol{Q}}_k(t)$, as defined in \eqref{eq24_f}. According to Lemmas \ref{lem1} and \ref{lem2}, $\delta\mathring{\boldsymbol{Q}}_k(t)=\pm\boldsymbol{1}+\epsilon \boldsymbol{0}$ is equivalent to $\mathring{\boldsymbol{Q}}_d(t)=\pm\mathring{\boldsymbol{Q}}_k(t)$, which means that the actual physical attitude and position align with their desired counterparts. Contrary to the error representation given by $\mathring{\boldsymbol{Q}}_d(t)-\mathring{\boldsymbol{Q}}_k(t)$, the quantity $\delta\mathring{\boldsymbol{Q}}_k(t)$ defined in \eqref{eq24_f} naturally adheres to the principles of unit dual quaternions, offering a clear physical interpretation.
\end{rem}

\subsection{Control Objectives}\label{sec_Problem}


This subsection presents the control objectives for position-attitude tracking. We first state the necessary assumptions for the theoretical analysis and practical relevance, and then formulate the specific tracking objectives.

The following two assumptions are widely used in ILC literature.

\begin{ass}\cite{zml:21}
\label{ass1_dq}
The initial conditions satisfy
\[
\delta\mathring{\boldsymbol{Q}}_k(0)=\pm\boldsymbol{1}+\epsilon \boldsymbol{0},\quad 
\delta\mathring{\boldsymbol{\omega}}_k(0)=\boldsymbol{0}+\epsilon \boldsymbol{0}, \quad \forall k\in\mathbb{Z}_+.
\]
\end{ass}

\begin{ass}
\label{ass2_dq}\cite{gws:21}
The system parameters and disturbances are bounded, i.e.,
\beaN
m \leq \beta_m, \quad \|\boldsymbol{J}\| \leq \beta_J, \quad \|\boldsymbol{d}_k(t)\| \leq \beta_d(t), 
\\
\|\boldsymbol{s}_k(t)\| \leq \beta_s(t), \quad \forall k\in\mathbb{Z}_+, \quad \forall t\in[0,T],
\eeaN
where $\beta_m<\infty$, $\beta_J<\infty$, $\beta_d(t)\in \mathscr{L}_{\infty e}$, and $\beta_s(t)\in \mathscr{L}_{\infty e}$ are unknown.
\end{ass}

With these assumptions, the control objective of this work is summarized as follows. 

For a rigid body governed by \eqref{eq17_f} and \eqref{eq21_f} with unknown system parameters, a sequence of control inputs to achieve simultaneous, exact position-attitude tracking of a six-degree-of-freedom trajectory $\mathring{\boldsymbol{Q}}_d(t)$ under \(\mathscr{L}_{\infty e}\)-norm bounded inputs over finite time intervals. Specifically, we aim to develop a robust adaptive ILC method such that:

\begin{enumerate}
\item[(C1)] \textbf{Perfect position-attitude tracking} is achieved over finite time intervals:
\[
\lim_{k\rightarrow\infty}\mathring{\boldsymbol{Q}}_k(t)= \mathring{\boldsymbol{Q}}_d(t),\quad \forall t\in[0,T],
\]
which, based on \eqref{eq24_f}, is equivalent to
\begin{equation}
\label{eq28_f}
\left\{
\begin{aligned}
\lim_{k\rightarrow \infty}\|\delta{\boldsymbol{P}}_k(t)\|_2 &= 0, \\
\lim_{k\rightarrow \infty}\|\delta{\boldsymbol{q}}_k(t)\|_2 &= 0,
\end{aligned}
\right.
\end{equation}
for all \(t\in[0,T]\), according to Lemma \ref{lem1}, where \(\delta{\boldsymbol{q}}_k(t) = \operatorname{red}\left(\delta{\boldsymbol{Q}}_k(t)\right)\).

\item[(C2)] \textbf{Uniform boundedness} of the real and dual parts of \(\mathring{\boldsymbol{f}}_k(t)\), \(\delta\mathring{\boldsymbol{Q}}_k(t)\), and \(\delta\mathring{\boldsymbol{\omega}}_k(t)\) in the \(\mathscr{L}_{\infty e}\)-norm sense. This enhances robustness against overestimation of torque and force uncertainties, which can not be achieved by adaptive ILC.
\end{enumerate}

\section{Controller Design and Analysis}
\label{sec_controller}

This section presents the dual learning control design, followed by a rigorous analysis of convergence and closed-loop boundedness properties.

\subsection{Two-Loop ILC Design}

Before proceeding to controller design, it is crucial to introduce two additional lemmas.

\begin{lem}
\label{lem5}
Let $\mathring{\boldsymbol{x}},\,\mathring{\boldsymbol{y}}\in\mathbb{DR}^3$ be two dual vectors. Then,
\begin{align}
\label{eq29_f}
\operatorname{crs}\left(\mathring{\boldsymbol{x}},-\mathring{\boldsymbol{M}}\left(\mathring{\boldsymbol{y}}^\times\mathring{\boldsymbol{x}}\right)\right)
=
\operatorname{crs}\left(\mathring{\boldsymbol{x}},\mathring{\boldsymbol{y}}^\times\mathring{\boldsymbol{M}}\mathring{\boldsymbol{x}}\right),
\end{align}
where $\operatorname{crs}\left(\cdot\right)$ is defined in \eqref{eq9_f}, and $\mathring{\boldsymbol{M}}$ is defined in \eqref{eq22_f}.
\end{lem}
\begin{proof}
The proof is provided in Appendix \ref{appendix:lem5}.
\end{proof}

\begin{lem}
\label{lem6}
Let three dual quaternions be denoted by $\mathring{\boldsymbol{Q}}_1,\,\mathring{\boldsymbol{Q}}_2,\,\mathring{\boldsymbol{Q}}_3\in\mathbb{DR}^4$. Then,
\begin{align}
\label{eq30_f}
\operatorname{crs}\left(\mathring{\boldsymbol{Q}}_1\circ \mathring{\boldsymbol{Q}}_2,\mathring{\boldsymbol{Q}}_3\right)
=
\operatorname{crs}\left(\mathring{\boldsymbol{Q}}_2,\mathring{\boldsymbol{Q}}^*_1\circ \mathring{\boldsymbol{Q}}_3\right).
\end{align}
\end{lem}

\begin{proof}
The proof is provided in Appendix \ref{appendix:lem6}.
\end{proof}

\begin{rem}
Lemmas \ref{lem5} and \ref{lem6} play a pivotal role in simplifying the error dynamics and facilitating the derivation of the control law. Specifically, Lemma \ref{lem5} enables the manipulation of dual vector cross products involving operator $\frac{\text{d}}{\text{d}\epsilon}$, which commonly arise in the position-attitude motion for rigid bodies. Lemma \ref{lem6}, on the other hand, provides a dual quaternion analog of the adjoint property for the cross-product operator, allowing the transfer of a dual quaternion multiplication from one argument of the $\operatorname{crs}(\cdot,\cdot)$ operator to the other via conjugation.
\end{rem}

Furthermore, one function is defined as
\begin{align}
\label{eq31_f}
\operatorname{func}\left(\mathring{\boldsymbol{x}},\mathring{\boldsymbol{y}},\mathring{\boldsymbol{z}}\right)
\triangleq{}&
\left(\left\|\boldsymbol{y}_r^\times\boldsymbol{y}_c+\boldsymbol{z}_c\right\|_2+1\right) \operatorname{sgn}\left(\boldsymbol{x}_c\right)
\nonumber
\\
&+
\epsilon
\left(\|\boldsymbol{y}_r\|_2^2+\|\boldsymbol{z}_r\|_2+1\right) \operatorname{sgn}\left(\boldsymbol{x}_r\right)
\end{align}
\noindent for any $\mathring{\boldsymbol{x}}=\boldsymbol{x}_r+\epsilon\boldsymbol{x}_c\in\mathbb{DR}^3$, $\mathring{\boldsymbol{y}}=\boldsymbol{y}_r+\epsilon\boldsymbol{y}_c\in\mathbb{DR}^3$, and $\mathring{\boldsymbol{z}}=\boldsymbol{z}_r+\epsilon\boldsymbol{z}_c\in\mathbb{DR}^3$. This function is used to simplify the subsequent controller design.

Moreover, there exists a time sequence denoted by $0=h_0<h_1<h_2<\cdots<h_{s-1}<h_s=T$ for some $s\geqslant 1$. Additionally, an index can be established using
\begin{align}
\label{eq32_f}
p(t)
\triangleq
\min
\left\{
j\in \Z_{s}: h_{j} \geqslant t
\right\}.
\end{align}
According to the definition of $p(t)$ in \eqref{eq32_f}, we formulate an innovative mechanism referred to as segment-based dynamic projection, which is represented by $\operatorname{proj}\left(\cdot\right)$ and drives the resulting outputs to comply with a prescribed dynamic bound. For any $x(t)\in\R$, $\operatorname{proj}\left(\cdot\right)$ is designed as
\begin{align}
\label{eq33_f}
&\operatorname{proj}\left(x(t)\right)
\nonumber
\\
={}&
\begin{cases}
x(t),\,\text{if } x(t)>\max_{\tau\in(h_{p(t)-1},h_{p(t)}]}{x(\tau)}-k_c\\
\max_{\tau\in(h_{p(t)-1},h_{p(t)}]}{x(\tau)}-k_c,\,\text{otherwise}
\end{cases}
\end{align}
\noindent where $k_c> 0$ is an undetermined parameter that represents the width of the dynamic projection.

Based on \eqref{eq33_f}, the controller $\mathring{\boldsymbol{f}}_k(t)$ is proposed as
\begin{align}
\label{eq34_f}
\!\mathring{\boldsymbol{f}}_k(t)
={}&
-\hat\theta_k(t)
\operatorname{func}\left(\delta\mathring{\boldsymbol{\omega}}_k(t),\operatorname{tr}
\left(\mathring{\boldsymbol{\omega}}_{d}(t)\right),\operatorname{tr}\left(\dot{\mathring{\boldsymbol{\omega}}}_{d}(t)\right)\right)
\nonumber
\\
&
-
k_d \operatorname{exch}\left(\delta\mathring{\boldsymbol{\omega}}_k(t)\right)
-\epsilon k_p\operatorname{red}\left( \operatorname{real}\left(\delta\mathring{\boldsymbol{Q}}_k(t)\right)
\right)
\nonumber
\\
&
-k_p\operatorname{red}\left(
\operatorname{real}\left(\delta\mathring{\boldsymbol{Q}}^*_k(t)\right)\circ\operatorname{comp}\left(\delta\mathring{\boldsymbol{Q}}_k(t)\right)
\right)\!,
\end{align}
\noindent where operators $\operatorname{exch}\left(\cdot\right)$, $\operatorname{tr}(\cdot)$, and $\operatorname{func}\left(\cdot\right)$ are defined in \eqref{eq9_f}, \eqref{eq26_f}, and \eqref{eq31_f}, respectively, $k_p>0$ and $k_d>0$ are the control parameters, and $\hat{\theta}_k(t)$ is the estimate of
$$
\theta(t)
\triangleq
\sup_{\tau\in(h_{p(t)-1},h_{p(t)}]}
\left\{\beta_m,\beta_J,\beta_d(\tau),\beta_s(\tau)\right\}.
$$
The corresponding iterative updating law is proposed as
\begin{align}
\label{eq35_f}
&\hat{{\theta}}_{k}(t)
-
\operatorname{proj}\left(\hat{{\theta}}_{k-1}(t)\right)
\nonumber
\\
={}&
k_\theta
\operatorname{crs}\Big(\delta \mathring{\boldsymbol{\omega}}_k(t),
\operatorname{func}\Big(\delta\mathring{\boldsymbol{\omega}}_k(t),\operatorname{tr}
\left(\mathring{\boldsymbol{\omega}}_{d}(t)\right),
\nonumber
\\
&\quad\quad\quad\quad\quad\quad\quad\quad\quad
\operatorname{tr}\left(\dot{\mathring{\boldsymbol{\omega}}}_{d}(t)\right)\!\Big)
\Big),
\end{align}
\noindent where $\operatorname{proj}\left(\cdot\right)$ is defined in \eqref{eq33_f}, $k_\theta>0$ is the learning parameter, and $\hat{{\theta}}_{0}(t)=0$ holds for all $t\in[0,T]$.

\begin{rem}
Although the controller $\mathring{\boldsymbol{f}}_k(t)$ is able to handle the uncertainties, it is worth noting that in practical engineering applications, additional prior knowledge—such as gravitational forces, Coulomb friction, or known dynamic disturbances—is often incorporated through feedforward compensation to further improve control performance. Importantly, such enhancements do not affect the theoretical validity or implementation of the proposed controller, as they can be seamlessly integrated as additive terms within the control law. The robustness and stability properties established in the analysis remain preserved, provided the feedforward components are bounded and do not introduce unmodeled dynamics.
\end{rem}

\subsection{Convergence and Boundedness Analysis}

By employing the two-loop ILC, we can achieve both perfect tracking and boundedness objectives presented in Section \ref{sec_Problem}. The detailed results are stated in the following theorem.

\begin{thm}
\label{th1}
For the position-attitude motion of rigid bodies governed by \eqref{eq17_f} and \eqref{eq21_f} under Assumptions \ref{ass1_dq} and \ref{ass2_dq}, let the dual learning control scheme consisting of \eqref{eq34_f} and \eqref{eq35_f} be applied. Then, the objectives outlined in (a1) and (a2), namely, perfect position-attitude tracking objective and the boundedness objective in $\mathscr{L}_{\mathrm{\infty e}}$-norm sense, are achieved.
\end{thm}

\begin{proof}
Initially, our examination focuses on the condition within $[0, h_1]$. A composite energy function is proposed as
\begin{align}
\label{eq36_f}
W_k(t)\triangleq
V_k(t)
+
\frac{1}{2 k_\theta}\int_{0}^{t}
\tilde{\theta}^2_{k}(\tau)
\diff \tau
\end{align}
\noindent with
\begin{align*}
V_k(t)
\triangleq{}&
\operatorname{crs}\left(\delta\mathring{\boldsymbol{Q}}_k(t)-\boldsymbol{1},\operatorname{exch}\left(\delta\mathring{\boldsymbol{Q}}_k(t)-\boldsymbol{1}\right)\right)
\nonumber
\\
{}&+
\frac{1}{2} \operatorname{crs}\left(\delta \mathring{\boldsymbol{\omega}}_k(t),\mathring{\boldsymbol{M}}\delta \mathring{\boldsymbol{\omega}}_k(t)\right),
\end{align*}
\noindent where $\tilde{\theta}_k(\tau)\triangleq \theta(t) - \hat{\theta}_k(\tau)$. Note that $\theta(t)$ remains constant over $[0,h_1]$ according to the definition of $\theta(t)$.

Additionally, the difference between $W_k(t)$ and $W_{k-1}(t)$ can be denoted by
\begin{align}
\label{eq37_f}
\Delta W_k(t)\triangleq W_k(t)-W_{k-1}(t).
\end{align}
\noindent Furthermore, we have
\begin{align}
\label{eq38_f}
\Delta W_k(t)
={}&
\frac{1}{2 k_\theta}\int_{0}^{t}
\left(
\tilde{\theta}^2_{k}(\tau)
-
\tilde{\theta}^2_{k-1}(\tau)
\right)
\diff \tau
\nonumber
\\
&
+
\int_{0}^{t}
\dot{V}_k(\tau)
\diff \tau
-
V_{k-1}(t),
\end{align}
\noindent where $\dot{V}_k(t)$ is the derivative of $V_k(t)$. Note that $V_k(0)$ does not occur in \eqref{eq38_f} because of Assumption \ref{ass1_dq}, which states that $V_k(0)=0,\,\forall k\in\Z_{+}$ holds.

To begin our analysis, we have the expression of $\dot{V}_k(\tau)$ as
\begin{align}
\label{eq39_f}
\dot{V}_k(\tau)
={}&
2k_p\operatorname{crs}\left(\delta\dot{\mathring{\boldsymbol{Q}}}_k(\tau),\operatorname{exch}\left(\delta\mathring{\boldsymbol{Q}}_k(\tau)-\boldsymbol{1}\right)\right)
\nonumber
\\
&
+
\operatorname{crs}\left(\delta \mathring{\boldsymbol{\omega}}_k(\tau),\mathring{\boldsymbol{M}}\delta \dot{\mathring{\boldsymbol{\omega}}}_k(\tau)\right).
\end{align}
\noindent The first term on the right-hand side of \eqref{eq39_f} is rewritten as
\begin{align*}
&2k_p\operatorname{crs}\left(\delta\dot{\mathring{\boldsymbol{Q}}}_k(\tau),\operatorname{exch}\left(\delta\mathring{\boldsymbol{Q}}_k(\tau)-\boldsymbol{1}\right)\right)
\nonumber
\\
={}&
k_p\operatorname{crs}\left(
\delta\mathring{\boldsymbol{Q}}_k(\tau)
\circ
\operatorname{aug}\left(
\delta\mathring{\boldsymbol{\omega}}_k(\tau)\right),\operatorname{exch}\left(\delta\mathring{\boldsymbol{Q}}_k(\tau)\right)\right)
\nonumber
\\
{}&-k_p
\operatorname{crs}\left(
\delta\mathring{\boldsymbol{Q}}_k(\tau)
\circ
\operatorname{aug}\left(
\delta\mathring{\boldsymbol{\omega}}_k(\tau)\right),\boldsymbol{0}+\epsilon\boldsymbol{1}\right)
\end{align*}
\noindent according to \eqref{eq25_f} and the definition of $\operatorname{exch}\left(\cdot\right)$ in \eqref{eq9_f}. By applying Lemma \ref{lem6} and considering the definition of $\operatorname{crs}\left(\cdot\right)$ in \eqref{eq9_f}, we further have
\begin{align}
\label{eq40_f}
&2k_p\operatorname{crs}\left(\delta\dot{\mathring{\boldsymbol{Q}}}_k(\tau),\operatorname{exch}\left(\delta\mathring{\boldsymbol{Q}}_k(\tau)-\boldsymbol{1}\right)\right)
\nonumber
\\
={}&
k_p\operatorname{crs}\left(
\operatorname{aug}\left(
\delta\mathring{\boldsymbol{\omega}}_k(\tau)\right),
\delta\mathring{\boldsymbol{Q}}^*_k(\tau)
\circ
\operatorname{exch}\left(\delta\mathring{\boldsymbol{Q}}_k(\tau)\right)\right)
\nonumber
\\
{}&-k_p
\operatorname{crs}\left(
\operatorname{aug}\left(
\delta\mathring{\boldsymbol{\omega}}_k(\tau)\right),
\delta\mathring{\boldsymbol{Q}}^*_k(\tau)
\circ\left(\boldsymbol{0}+\epsilon\boldsymbol{1}\right)\right)
\nonumber
\\
={}&
k_p\operatorname{crs}\left(
\delta\mathring{\boldsymbol{\omega}}_k(\tau),
\operatorname{red}\left(
\delta\mathring{\boldsymbol{Q}}^*_k(\tau)
\circ
\operatorname{exch}\left(\delta\mathring{\boldsymbol{Q}}_k(\tau)\right)
\right)
\right)
\nonumber
\\
{}&-k_p
\operatorname{crs}\left(
\delta\mathring{\boldsymbol{\omega}}_k(\tau),
\operatorname{red}\left(
\delta\mathring{\boldsymbol{Q}}^*_k(\tau)
\circ\left(\boldsymbol{0}+\epsilon\boldsymbol{1}\right)\right)
\right).
\end{align}
Based on the distributive property of the dual quaternion multiplication, we have 
\begin{align}
\label{eq41_f}
&\operatorname{red}\left(
\delta\mathring{\boldsymbol{Q}}^*_k(\tau)
\circ
\operatorname{exch}\left(\delta\mathring{\boldsymbol{Q}}_k(\tau)\right)
\right)
\nonumber
\\
={}&
\operatorname{red}
\left(
\operatorname{real}\left(
\delta\mathring{\boldsymbol{Q}}^*_k(\tau)
\right)
\circ
\operatorname{comp}\left(\delta\mathring{\boldsymbol{Q}}_k(\tau)\right)
\right)
\nonumber
\\
{}&
+\epsilon\operatorname{red}
\left(
\operatorname{real}\left(
\delta\mathring{\boldsymbol{Q}}^*_k(\tau)
\right)
\circ
\operatorname{real}\left(\delta\mathring{\boldsymbol{Q}}_k(\tau)\right)
\right)
\nonumber
\\
{}&
+\epsilon\operatorname{red}
\left(
\operatorname{comp}\left(
\delta\mathring{\boldsymbol{Q}}^*_k(\tau)
\right)
\circ
\operatorname{comp}\left(\delta\mathring{\boldsymbol{Q}}_k(\tau)\right)
\right)
\nonumber
\\
={}&
\operatorname{red}
\left(
\operatorname{real}\left(
\delta\mathring{\boldsymbol{Q}}^*_k(\tau)
\right)
\circ
\operatorname{comp}\left(\delta\mathring{\boldsymbol{Q}}_k(\tau)\right)
\right),
\end{align}
\noindent where \eqref{eq8_f} and \eqref{eq13_f} are applied. Adopting a similar analysis approach, we obtain
\begin{align}
\label{eq42_f}
\operatorname{red}\left(
\delta\mathring{\boldsymbol{Q}}^*_k(\tau)
\circ\left(\boldsymbol{0}+\epsilon\boldsymbol{1}\right)\right)
={}&
\operatorname{red}\left(
\operatorname{real}\left(
\delta\mathring{\boldsymbol{Q}}^*_k(\tau)
\right)
\circ
\left(\epsilon\boldsymbol{1}\right)\right)
\nonumber
\\
={}&
-\operatorname{red}\left(
\epsilon
\operatorname{real}\left(
\delta\mathring{\boldsymbol{Q}}_k(\tau)
\right)\right)
\end{align}
\noindent according to the definition of $\boldsymbol{1}$. Substituting \eqref{eq41_f} and \eqref{eq42_f} into \eqref{eq40_f}, it is obtained that
\begin{align}
\label{eq43_f}
&2k_p\operatorname{crs}\left(\delta\dot{\mathring{\boldsymbol{Q}}}_k(\tau),\operatorname{exch}\left(\delta\mathring{\boldsymbol{Q}}_k(\tau)-\boldsymbol{1}\right)\right)
\nonumber
\\
={}&
k_p\operatorname{crs}\Big(
\delta\mathring{\boldsymbol{\omega}}_k(\tau),\operatorname{red}\Big(
\operatorname{real}\left(\delta\mathring{\boldsymbol{Q}}^*_k(\tau)\right)\circ\operatorname{comp}\left(\delta\mathring{\boldsymbol{Q}}_k(\tau)\right)
\nonumber
\\
&\quad\quad\quad\quad\quad\quad\quad\quad\quad
+\epsilon \operatorname{real}\left(\delta\mathring{\boldsymbol{Q}}_k(\tau)\right)
\Big)\,\Big),
\end{align}
\noindent where \eqref{eq13_f} is applied. 

\par Next, let us study the second term on the right-hand side of \eqref{eq39_f} with the substitution of \eqref{eq27_f}. According to Lemma \ref{lem5}, we have
\begin{align}
\label{eq44_f}
&\operatorname{crs}\left(\delta\mathring{\boldsymbol{\omega}}_k(\tau),-\mathring{\boldsymbol{M}}
\left(\mathring{\boldsymbol{\omega}}^\times_k(\tau)
\delta{\mathring{\boldsymbol{\omega}}}_k(\tau)
\right)\right)
\nonumber
\\
=&{}
\operatorname{crs}\left(\delta\mathring{\boldsymbol{\omega}}_k(\tau),\mathring{\boldsymbol{\omega}}^\times_k(\tau) \mathring{\boldsymbol{M}}
\delta\mathring{\boldsymbol{\omega}}_k(\tau)\right).
\end{align}
Based on \eqref{eq26_f}, it is obtained that
\begin{align}
\label{eq45_f}
&\mathring{\boldsymbol{\omega}}^\times_k(\tau) \mathring{\boldsymbol{M}}
\delta\mathring{\boldsymbol{\omega}}_k(\tau)
-
\mathring{\boldsymbol{\omega}}^\times_k(\tau) \mathring{\boldsymbol{M}}\mathring{\boldsymbol{\omega}}_k(\tau)
\nonumber
\\
={}&
-\mathring{\boldsymbol{\omega}}^\times_k(\tau) \mathring{\boldsymbol{M}}
\operatorname{tr}\left(\mathring{\boldsymbol{\omega}}_{d}(\tau)\right).
\end{align}
\noindent Additionally, we have
\begin{align}
\label{eq46_f}
&
\operatorname{crs}\left(\delta\mathring{\boldsymbol{\omega}}_k(\tau),-\mathring{\boldsymbol{\omega}}^\times_k(\tau) \mathring{\boldsymbol{M}}
\operatorname{tr}
\left(\mathring{\boldsymbol{\omega}}_{d}(\tau)\right)
\right)
\nonumber
\\
=&
-\operatorname{crs}\left(\delta\mathring{\boldsymbol{\omega}}_k(\tau),\left(\operatorname{tr}
\left(\mathring{\boldsymbol{\omega}}_{d}(\tau)\right)\right)^\times\mathring{\boldsymbol{M}}
\operatorname{tr}
\left(\mathring{\boldsymbol{\omega}}_{d}(\tau)\right)\right),
\end{align}
\noindent where $\operatorname{crs}\left(\delta\mathring{\boldsymbol{\omega}}_k(\tau),-\delta\mathring{\boldsymbol{\omega}}^\times_k(\tau)\mathring{\boldsymbol{M}}
\operatorname{tr}
\left(\mathring{\boldsymbol{\omega}}_{d}(\tau)\right)\right)=0$ is applied. Substituting \eqref{eq43_f}-\eqref{eq46_f} into \eqref{eq39_f}, we have
\begin{align}
\label{eq47_f}
{}&\dot{V}_k(\tau)
\nonumber
\\
={}&-\operatorname{crs}\left(\delta\mathring{\boldsymbol{\omega}}_k(\tau),\left(\operatorname{tr}
\left(\mathring{\boldsymbol{\omega}}_{d}(\tau)\right)\right)^\times\mathring{\boldsymbol{M}}
\operatorname{tr}
\left(\mathring{\boldsymbol{\omega}}_{d}(\tau)\right)\right)
\nonumber
\\
&-\operatorname{crs}\left(\delta\mathring{\boldsymbol{\omega}}_k(\tau),\mathring{\boldsymbol{M}}
\operatorname{tr}\left(\dot{\mathring{\boldsymbol{\omega}}}_{d}(\tau)\right)\right)
\nonumber
\\
&
+
\operatorname{crs}\left(\delta\mathring{\boldsymbol{\omega}}_k(\tau),\mathring{\boldsymbol{d}}_k(\tau)\right)
+
\operatorname{crs}\left(\delta\mathring{\boldsymbol{\omega}}_k(\tau),\mathring{\boldsymbol{f}}_k(\tau)\right)
\nonumber
\\
{}&+k_p
\operatorname{crs}\Big(
\delta\mathring{\boldsymbol{\omega}}_k(\tau),\operatorname{red}\Big(
\operatorname{real}\left(\delta\mathring{\boldsymbol{Q}}^*_k(\tau)\right)\circ\operatorname{comp}\left(\delta\mathring{\boldsymbol{Q}}_k(\tau)\right)
\nonumber
\\
&\quad\quad\quad\quad\quad\quad\quad\quad\quad
+\epsilon \operatorname{real}\left(\delta\mathring{\boldsymbol{Q}}_k(\tau)\right)
\Big)\Big).
\end{align}
Considering the definition of $\theta(t)$ and substituting \eqref{eq34_f} into \eqref{eq47_f}, we have, for all $k\in\Z_{+}$,
\begin{align}
\label{eq48_f}
\dot{V}_k(\tau)\leqslant{}&
-
\operatorname{crs}\left(\delta\mathring{\boldsymbol{\omega}}_k(\tau),k_d\operatorname{exch}\left(\delta\mathring{\boldsymbol{\omega}}_k(\tau)\right)\right)
\nonumber
\\
&
+
\tilde\theta_k(\tau)\operatorname{crs}\Big(\delta\mathring{\boldsymbol{\omega}}_k(\tau),
\operatorname{func}\Big(\delta\mathring{\boldsymbol{\omega}}_k(\tau),\operatorname{tr}
\left(\mathring{\boldsymbol{\omega}}_{d}(\tau)\right),
\nonumber
\\
&\quad\quad\quad\quad\quad\quad\quad\quad\quad\quad\quad~~
\operatorname{tr}\left(\dot{\mathring{\boldsymbol{\omega}}}_{d}(\tau)\right)\Big)\Big).
\end{align}

Drawing on the previous findings, we conduct the subsequent analysis.

\textbf{\textit{Step 1:}} Note the fact that $\operatorname{proj}\left(\hat{{\theta}}_{0}(t)\right)=0,\,\forall t\in[0,h_1]$ according to \eqref{eq33_f}, it is not hard to verify 
\begin{align}
\label{eq49_f}
W_0(t)\leqslant \beta_{W_0},\,\forall t\in[0,h_1],
\end{align}
\noindent where $\beta_{W_0}<\infty$
based on \eqref{eq36_f}, \eqref{eq48_f}, and Assumption \ref{ass2_dq}. The detailed derivation of \eqref{eq49_f} is omitted.

\textbf{\textit{Step 2:}} Let us discuss the evolution of $V_k(t)$ along the iteration axis, which can be divided into two cases.

\textit{Case 1:} If $\hat{\theta}_{k_0}(t)\leqslant \theta(t)+k_c,\,\forall t\in[0,h_1]$ holds for some ${k_0}<\infty$, where $k_c$ is defined in \eqref{eq33_f}, it is obtained that
$\max_{t\in[0,h_1]}{\hat{{\theta}}_{k_0}(t)}-k_c\leqslant \theta(t)$. Taking into account the definition of $\operatorname{proj}\left(\cdot\right)$, we can derive $\left|\theta(t)-\operatorname{proj}\left(\hat{{\theta}}_{k_0}(t)\right)\right|
\leqslant
\left|\theta(t)-\hat{{\theta}}_{k_0}(t)\right|$. By utilizing \eqref{eq35_f}, we further have 
\begin{align}
\label{eq50_f}
\left|\theta(t)-\operatorname{proj}\left(\hat{{\theta}}_{k}(t)\right)\right|
\leqslant
\left|\theta(t)-\hat{{\theta}}_{k}(t)\right|,\,\forall k\leqslant k_0.
\end{align}
\noindent Therefore, it is obtained that, for all $k\leqslant k_0+1$,
\begin{align}
\label{eq51_f}
&
\frac{1}{2k_\theta}\left(
\tilde{\theta}^2_{k}(\tau)
-\tilde{\theta}^2_{k-1}(\tau)
\right)
\nonumber
\\
\leqslant\,&
\frac{1}{2k_\theta}\left(
\tilde{\theta}^2_{k}(\tau)
-
\left(\theta(t)-\operatorname{proj}\left(\hat{{\theta}}_{k-1}(t)\right)\right)^2
\right)
\nonumber
\\
\leqslant\,&
-\frac{1}{k_\theta}
\left(\hat{\theta}_k(\tau) - \operatorname{proj}\left(\hat{{\theta}}_{k-1}(t)\right)\right)\tilde{\theta}_k(\tau),
\end{align}
\noindent where \eqref{eq35_f} is once again employed. Substituting \eqref{eq48_f} and \eqref{eq51_f} into \eqref{eq38_f} leads to
\begin{align}
\label{eq52_f}
\Delta W_k(t)
\leqslant{}&
-\int_{0}^{t}
\operatorname{crs}\left(\delta\mathring{\boldsymbol{\omega}}_k(\tau),k_d\operatorname{exch}\left(\delta\mathring{\boldsymbol{\omega}}_k(\tau)\right)\right)
\diff \tau
\nonumber
\\
{}&
-V_{k-1}(t)
,\,\forall k\leqslant k_0+1,\,\forall t\in[0,h_1],
\end{align}
\noindent which implies that $\Delta W_k(t)
\leqslant 0,\,\forall k\leqslant k_0+1$. Considering both \eqref{eq49_f} and \eqref{eq52_f}, we obtain
$
W_k(t)\leqslant \beta_{W_0},\,\forall k\leqslant k_0+1,\,\forall t\in[0,h_1]
$, which implies the boundedness of the real and imaginary parts of $\delta\mathring{\boldsymbol{Q}}_k(t)$ and $\delta \mathring{\boldsymbol{\omega}}_k(t)$ for all $k\leqslant k_0+1$ and all $t\in[0,h_1]$. This fact further shows
\begin{align}
\label{eq53_f}
\hat{\theta}_{k_0+1}(t)<\infty,\,\forall t\in[0,h_1]
\end{align} according to \eqref{eq35_f} and $\hat{\theta}_{k_0}(t)\leqslant \theta(t)+k_c,\,\forall t\in[0,h_1]$.

\textit{Case 2:} If $\hat{\theta}_{k_1}(t)> \theta(t)+k_c,\,\exists t\in[0,h_1]$ holds for some ${k_1}<\infty$, it is not hard to derive that $\tilde{{\theta}}_{k}(t)\leqslant 0,\,\forall k\geqslant {k_1+1},\,\forall t\in[0,h_1]$ holds according to \eqref{eq35_f}. As a result, $\dot{V}_k(t)\leqslant 0,\,\forall k\geqslant k_1+1,\,\forall t\in[0,h_1]$ holds based on \eqref{eq48_f}, which, together with Assumption \ref{ass1_dq}, implies that perfect tracking is achieved, namely, $\delta\mathring{\boldsymbol{Q}}_k(t)= \boldsymbol{1}+\epsilon \boldsymbol{0}$ and $\delta\mathring{\boldsymbol{\omega}}_k(t)= \boldsymbol{0}+\epsilon \boldsymbol{0}$ for all $k\geqslant k_1+1$ and all $t\in[0,h_1]$. Additionally, $\hat{{\theta}}_{k}(t)
=
\operatorname{proj}\left(\hat{{\theta}}_{k_1}(t)\right),\,\forall k\geqslant k_1+1,\,\forall t\in[0,h_1]$ holds according to \eqref{eq35_f}.

\textbf{\textit{Step 3:}} We establish the convergence and boundedness of both the tracking errors and estimates. Much like in Step 2, we also divide the discussion into two cases.

\textit{Case 1:} If $\hat{\theta}_{k}(t)\leqslant \theta(t)+k_c,\,\forall k\in\Z_{+},\,\forall t\in[0,h_1]$ holds, then \eqref{eq52_f} holds for all $k\in\Z_{+}$ and all $t\in[0,h_1]$. This fact, along with \eqref{eq37_f} and the positiveness of $V_k(t)$, implies
\begin{align}
\label{eq54_f}
\sum_{k=0}^{\infty} V_k(t)<\infty,\,\forall t\in[0,h_1]
\end{align}
and
\begin{align}
\label{eq55_f}
\lim\limits_{k\rightarrow \infty} V_k(t)=0,\,\forall t\in[0,h_1].
\end{align}

\textit{Case 2:} If $\hat{\theta}_{k}(t)\leqslant \theta(t)+k_c,\,\forall k\in\Z_{+},\,\forall t\in[0,h_1]$ is not satisfied, then \eqref{eq54_f} and \eqref{eq55_f} are also obtained according to the analysis in Step 2. Additionally, based on \eqref{eq53_f} and case 2 in Step 2, $\hat{\theta}$ is also bounded. 

Based on the above discussion, over $[0,h_1]$, perfect tracking is achieved, and the uniform boundedness of the real and dual parts of $\mathring{\boldsymbol{f}}_k(t)$, $\delta\mathring{\boldsymbol{Q}}_k(t)$, and $\delta\mathring{\boldsymbol{\omega}}_k(t)$ can be ensured in the $\mathscr{L}_{\mathrm{\infty e}}$-norm sense, which can not be achieved by adaptive ILC.

Subsequently, one can similarly and readily complete the analysis of perfect tracking performance and the uniform boundedness for the entire interval $(h_1,h_2]$ using \eqref{eq54_f}, even though $V_k(h_1)$ may not always be zero. Thus, we are able to obtain perfect tracking over $[0,T]$ and the $\mathscr{L}_{\mathrm{\infty e}}$-norm boundedness of all states, inputs, and estimates. The proof of Theorem \ref{th1} is concluded.
\end{proof}

\begin{rem}
In fact, $\mathring{\boldsymbol{f}}_k(t)+\mathring{\boldsymbol{d}}_k(t)$ encompasses all the forces and torques experienced by a rigid body. Hence, even though we haven't explicitly accounted for gravity in our previous analysis and design, it can be inferred that gravity has been pre-compensated nominally, and the remaining error term is encompassed within $\mathring{\boldsymbol{d}}_k(t)$.
\end{rem}

To mitigate potential overestimation in practical implementations, a saturation bound $k_l > 0$ can be incorporated into the iterative updating law \eqref{eq35_f}, yielding the modified scheme
\begin{align}
\label{eq56_f_add1}
&\hat{{\theta}}_{k}(t)
-
\operatorname{proj}\left(\hat{{\theta}}_{k-1}(t)\right)
\nonumber
\\
={}&
\min\Big\{k_l,
k_\theta
\operatorname{crs}\Big(\delta \mathring{\boldsymbol{\omega}}_k(t),
\operatorname{func}\Big(\delta\mathring{\boldsymbol{\omega}}_k(t),\operatorname{tr}
\left(\mathring{\boldsymbol{\omega}}_{d}(t)\right),
\nonumber
\\
&\quad\quad\quad\quad\quad\quad\quad\quad\quad\quad\quad\quad\quad
\operatorname{tr}\left(\dot{\mathring{\boldsymbol{\omega}}}_{d}(t)\right)\Big)
\Big)
\Big\},
\end{align}
where all variables follow the definitions previously given. The corresponding convergence property is formally stated in the following theorem.

\begin{thm}
\label{thm2}
For the position-attitude motion of rigid bodies governed by \eqref{eq17_f} and \eqref{eq21_f} under Assumptions \ref{ass1_dq} and \ref{ass2_dq}, let the dual learning control scheme consisting of \eqref{eq34_f} and \eqref{eq56_f_add1} be applied. Then, the objectives outlined in (a1) and (a2), namely, perfect position-attitude tracking objective and the boundedness objective in $\mathscr{L}_{\mathrm{\infty e}}$-norm sense, are achieved.
\end{thm}

\begin{proof}
The proof of Theorem \ref{thm2} follows a similar line to that of Theorem \ref{th1}. It is worth mentioning that the non-increasing property of the composite energy function is not always maintained due to $k_l$. However, this doesn't significantly complicate the analysis. This is because $k_l<
k_\theta
\operatorname{crs}\left(\delta \mathring{\boldsymbol{\omega}}_k(t),
\operatorname{func}\left(\delta\mathring{\boldsymbol{\omega}}_k(t),\operatorname{tr}
\left(\mathring{\boldsymbol{\omega}}_{d}(t)\right),\operatorname{tr}\left(\dot{\mathring{\boldsymbol{\omega}}}_{d}(t)\right)\right)
\right)$ holds for only a finite number of iterations for some $t$. Otherwise, it would contradict the boundedness property of the estimates ensured by the dynamic projection mechanism. Therefore, we omit further analysis in this regard.
\end{proof}

\begin{rem}
Despite the dependence of the proposed adaptive ILC on a special imaginary unit $\epsilon$ for design and analysis, the calculated force and torque commands are transmitted to the actuator without involving any dual numbers. Therefore, in the controller settings of practical applications or numerical simulations, the complex numbers can be avoided, making it easier for practitioners to operate the proposed adaptive ILC.
\end{rem}

\begin{figure*}[!htbp]
    \centering
    \begin{minipage}[b]{0.48\textwidth}
        \centering
        \includegraphics[width=\linewidth]{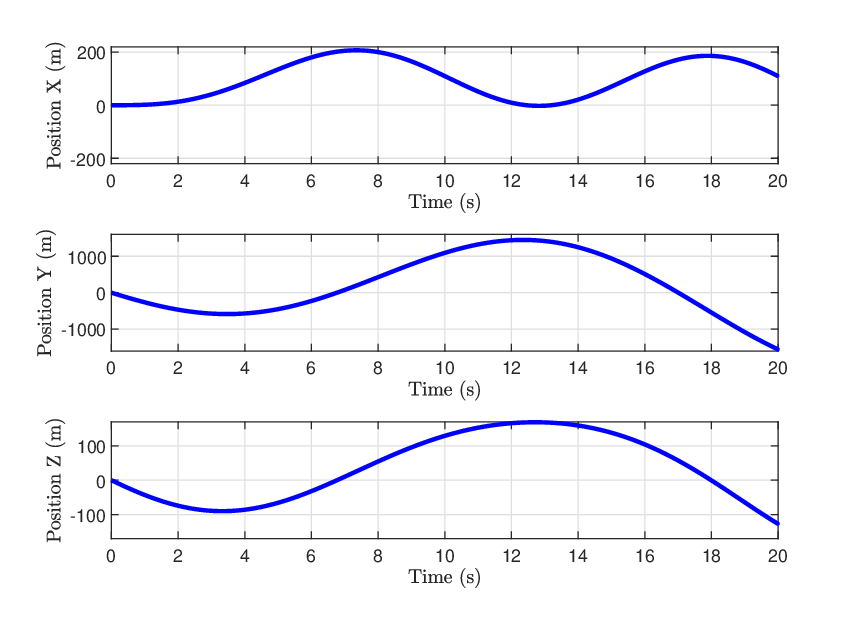}
        \caption{Position errors along the time axis for $k=0$.}
        \label{fig5}
    \end{minipage}
    \hfill
    \begin{minipage}[b]{0.48\textwidth}
        \centering
        \includegraphics[width=\linewidth]{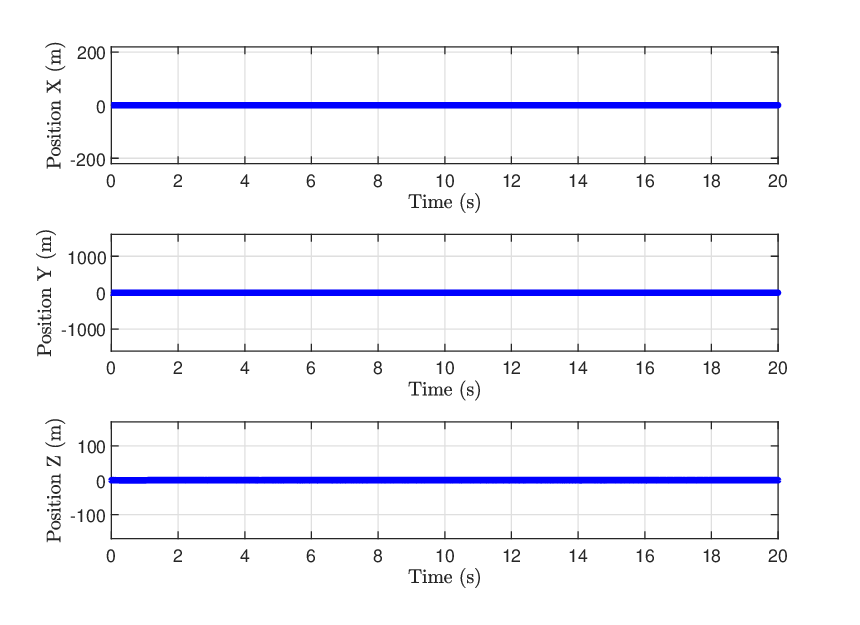}
        \caption{Position errors along the time axis for $k=30$.}
        \label{fig6}
    \end{minipage}
    
    \vspace{1em} 
    
    \begin{minipage}[b]{0.48\textwidth}
        \centering
        \includegraphics[width=\linewidth]{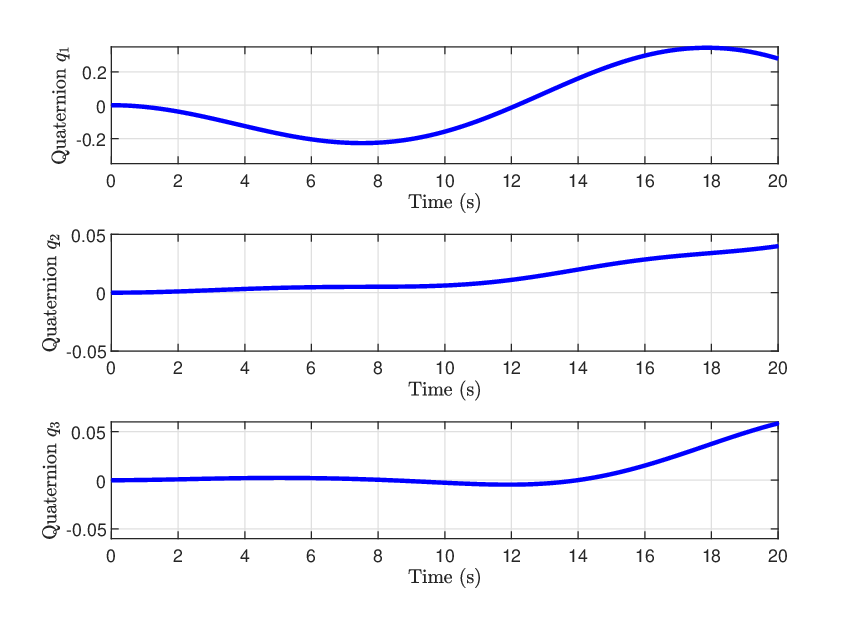}
        \caption{Attitude errors along the time axis for $k=0$.}
        \label{fig7}
    \end{minipage}
    \hfill
    \begin{minipage}[b]{0.48\textwidth}
        \centering
        \includegraphics[width=\linewidth]{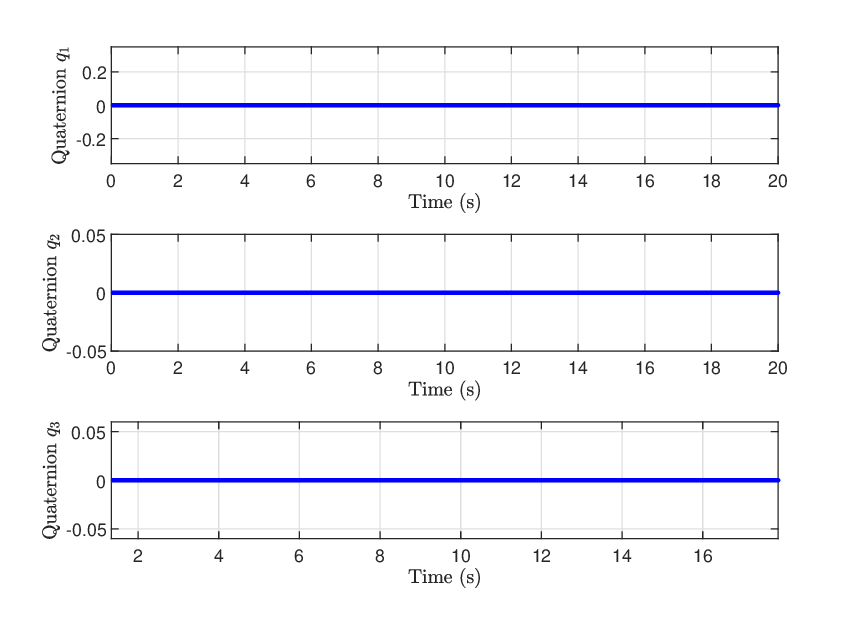}
        \caption{Attitude errors along the time axis for $k=30$.}
        \label{fig8}
    \end{minipage}
\end{figure*}

\section{Simulation Verification}\label{sec_Sim}

In this section, we conduct the numerical simulation to evaluate the effectiveness of the proposed two-loop adaptive ILC for precise six-degree-of-freedom trajectory tracking in the proximity operations of rigid satellites. The simulations are designed under realistic conditions, incorporating predefined position and attitude trajectories alongside various uncertainties and disturbances.

A rigid satellite is utilized to repetitively track a predefined six-degree-of-freedom trajectory. We apply the position-attitude trajectory described by $\mathring{\boldsymbol{Q}}_d(0)$ and $\mathring{\boldsymbol{\omega}}_d(t)$, that is,
\begin{align*}
\mathring{\boldsymbol{Q}}_d(0)\!=\!\underbrace{\left[\begin{array}{r}
0.7055\\
0.0471\\
-0.7055\\
-0.0471
\end{array}\right]}_{\text{unit: }\mathrm{1}}
\!+
\epsilon
\frac{1}{2}
\underbrace{\left[\begin{array}{r}
0.7055\\
0.0471\\
-0.7055\\
-0.0471
\end{array}\right]}_{\text{unit: }\mathrm{1}}
\circ
\underbrace{\left[\begin{array}{r}
0\\
0\\
0\\
-6778200
\end{array}\right]}_{\text{unit: }\mathrm{m}}
\end{align*}
\noindent and
\begin{align*}
	\mathring{\boldsymbol{\omega}}_d(t)\!=\!
	\underbrace{\left[\begin{array}{r}
			\frac{\pi}{8} \cdot \frac{\pi}{10}\sin(\frac{\pi}{10}t)                  \\
			-\omega^{\prime}\cos(-\frac{\pi}{8}\cos(\frac{\pi}{10}t)+\frac{\pi}{8}) \\
			\omega^{\prime}\sin(-\frac{\pi}{8}\cos(\frac{\pi}{10}t)+\frac{\pi}{8})
		\end{array}\right]}_{\text{unit: } \mathrm{rad/s}}
		\!+
		\epsilon
		\underbrace{\left[\begin{array}{r}
					7668.5229 \\
					0 \\
					0
				\end{array}\right]}_{\text{unit: }\mathrm{m/s}}\!,
\end{align*}
\noindent where $\omega^{\prime}=0.0011~ \mathrm{rad} / \mathrm{s}$. The operation length is $T=20~ \mathrm{s}$, which is divided into $s=2$ sequences, the real and nominal mass are $m=19~\mathrm{kg}$ and $m_n=20~\mathrm{kg}$, the real and nominal inertial matrix are \begin{align*}
\boldsymbol{J} = \left[\begin{array}{rrr}
			12 & 1  & 1  \\
			1  & 10 & 2  \\
			1  & 2  & 10
		\end{array}\right] ~{\mathrm{kg} \cdot \mathrm{m}^{2}}
\end{align*}
and
\begin{align*}
\boldsymbol{J}_n = \left[\begin{array}{rrr}
			20 & 2  & 1  \\
			2  & 15 & 3  \\
			1  & 3  & 15
		\end{array}\right] ~{\mathrm{kg} \cdot \mathrm{m}^{2}}.
\end{align*}
\noindent According to Assumption \ref{ass1_dq}, $\mathring{\boldsymbol{Q}}_k(0)=\mathring{\boldsymbol{Q}}_d(0)$ and $\mathring{\boldsymbol{\omega}}_k(0)=\mathring{\boldsymbol{\omega}}_d(0)$ are given for all $k\in\Z_{+}$. Additionally, the disturbance torque $\boldsymbol{s}_k(t)$ is supposed to exhibit a sinusoidal waveform characterized by periods $400$, $500$, and $700$ $\mathrm{s}$ and magnitudes of $0.1$, $0.05$, and $0.08$ $\mathrm{Nm}$, respectively. The disturbance force $\boldsymbol{d}_k(t)$ contains $\boldsymbol{d}^{'}_k(t)$ and $\boldsymbol{d}^{''}_k(t)$. Suppose that $\boldsymbol{d}^{'}_k(t)$ adheres to a comparable pattern, featuring periods of $100$, $200$, and $300$ $\mathrm{s}$, as well as magnitudes of $0.5$, $0.5$, and $0.5$ $\mathrm{N}$, respectively. The initial phases of $\boldsymbol{d}^{'}_k(t)$ are iteration-dependent and fall within the range of $[0, 0.1\pi]$. Additionally, $\boldsymbol{d}^{''}_k(t)$ is caused by universal gravitation.

The control parameters are set as $k_p=1$, $k_d=1$, and $k_c=0.01$, with $k_\theta=0.002$ for Theorem \ref{th1} and $k_l=0.02$ for Theorem \ref{thm2}. Both the simulation and control frequencies are fixed at $1000~\mathrm{Hz}$. The total number of iterations is preset to $30$, with the number of segments $s=200$. Note that in numerical simulations or practical applications, it is important to ensure that parameters are positive and not overly large to adhere to the limits set by the control frequency. Due to the similarity between Theorem \ref{th1} and Theorem \ref{thm2}, except for the convergence rate, only the results corresponding to Theorem \ref{thm2} are presented below.

The simulation results are illustrated in Figs. \ref{fig5}--\ref{fig4}. The position tracking errors at iteration $k=0$ and $k=30$ are respectively displayed in Figs. \ref{fig5} and \ref{fig6}, and the attitude tracking errors at iteration $k=0$ and $k=30$ are respectively displayed in Figs. \ref{fig7} and \ref{fig8}. Specifically, the evolution of the position/attitude tracking errors along the iteration axis is shown in Fig. \ref{fig3} and \ref{fig4}. At $k=0$, $\theta_{0}(t)=0,\forall t\in[0,T]$, indicating that the controller is of the PD type. If the controller parameters are poorly chosen, the maximum value of $\|\delta \boldsymbol{P}_k(t)\|_2$ for all $t\in[0,T]$ exceeds $1582~\mathrm{m}$, as shown in Fig. \ref{fig5}. As the number of iterations increases, the position tracking errors tend to zero, with the maximum value of $\|\delta \boldsymbol{P}_k(t)\|_2$ over the entire interval reducing to approximately $33~\mathrm{m}$. In Fig. \ref{fig4}, the evolution of attitude tracking errors exhibits a pattern reminiscent of the trends observed in position tracking errors in Fig. \ref{fig3}. Based on the Euler angle definition in \eqref{eq7_f_add1}, the maximum error Euler angle is found to be $42.3~\mathrm{deg}$ at iteration $k=0$ in \ref{fig7}, and $0.1~\mathrm{deg}$ at iteration $k=30$ in \ref{fig8}. These facts demonstrate a substantial improvement in position-attitude coupling tracking performance the entire $20~\mathrm{s}$ interval. Additionally, the maximum value of the estimate is $0.4981$ for $k=30$, and actually it remains almost unchanged starting from $k=25$. This fact shows the boundedness of estimates and inputs. 

In summary, the proposed ILC framework and the dynamic projection mechanism prove to be effective, successfully achieving the two stated objectives (a1) and (a2).

\section{Conclusions}
\label{sec_Conclusion}

In this paper, we have explored a segment-based two-loop adaptive ILC framework tailored for close-range manipulation of rigid bodies, specifically addressing the challenge of repetitively tracking six-degree-of-freedom position and attitude trajectories under time-iteration-dependent uncertainties. The resolution of this problem has been achieved through the utilization of dual numbers, specifically referred to as complex numbers of parabolic type, which enable the establishment of a two-loop adaptive ILC framework. Within this framework, a segment-based dynamic projection mechanism has been developed as a natural extension to guarantee the boundedness of estimates and control inputs. Comprehensive analyzes and simulations have demonstrated that our approach effectively manages various uncertainties and interactions between position and attitude, ensuring robust performance. The findings have confirmed the efficacy of the proposed control scheme in achieving the high-precision attitude-position tracking for complex manipulation tasks.

\appendices
\section{Proof of Lemma \ref{lem1}}
\label{appendix:lem1}
\begin{proof}
Using \eqref{eq13_f} and noticing ${\mathring{\boldsymbol{Q}}}_1,\,{\mathring{\boldsymbol{Q}}}_2\in\mathbb{DS}^3$, we have
\begin{align*}
{\mathring{\boldsymbol{Q}}}^*_1\circ{\mathring{\boldsymbol{Q}}}_2
\circ ({\mathring{\boldsymbol{Q}}}^*_1\circ{\mathring{\boldsymbol{Q}}}_2)^*
=
{\mathring{\boldsymbol{Q}}}^*_1\circ{\mathring{\boldsymbol{Q}}}_2
\circ
\mathring{\boldsymbol{Q}}^{*}_2
\circ
\mathring{\boldsymbol{Q}}_1
=
\boldsymbol{1}+\epsilon \boldsymbol{0},
\end{align*}
\noindent which confirms ${\mathring{\boldsymbol{Q}}}^*_1\circ{\mathring{\boldsymbol{Q}}}_2\in\mathbb{DS}^3$, as required.

\begin{figure*}[!htbp]
    \centering
    \begin{minipage}[b]{0.48\textwidth}
        \centering
        \includegraphics[width=\linewidth]{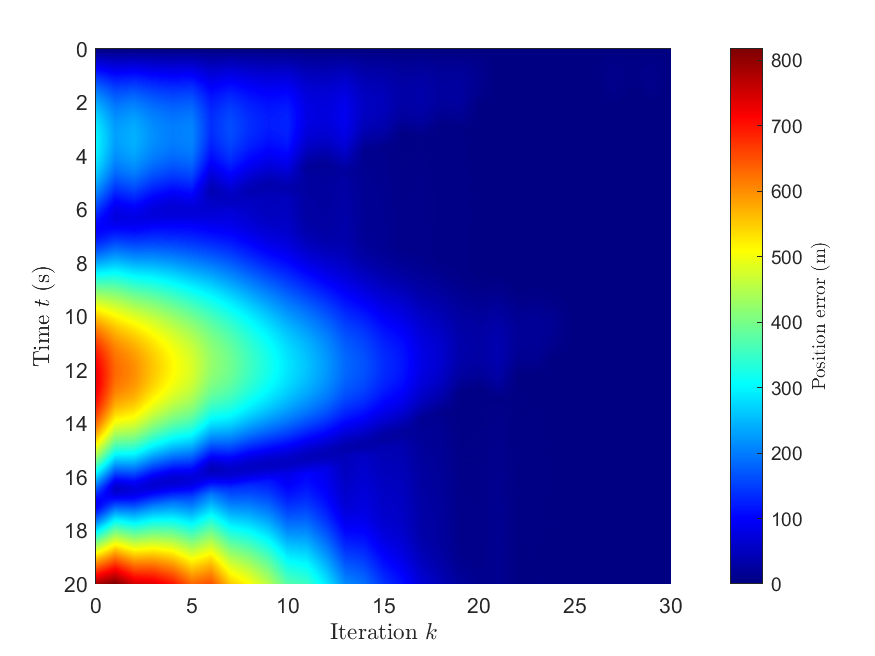}
        \caption{Position tracking performance along the iteration axis.}
        \label{fig3}
    \end{minipage}
    \hfill
    \begin{minipage}[b]{0.48\textwidth}
        \centering
        \includegraphics[width=\linewidth]{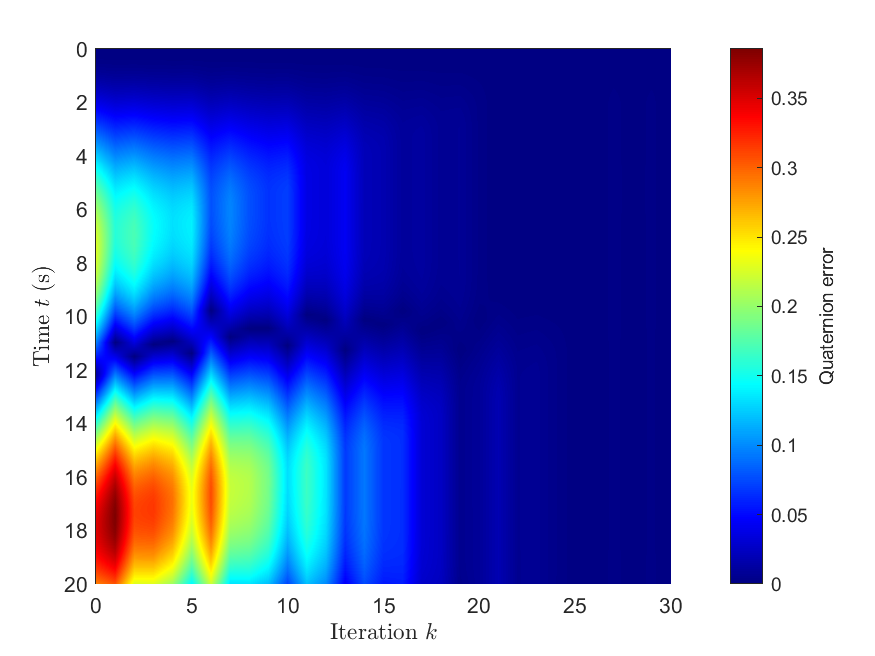}
        \caption{Attitude tracking performance along the iteration axis.}
        \label{fig4}
    \end{minipage}
\end{figure*}

Next, let us investigate the equivalence between ${\mathring{\boldsymbol{Q}}}^*_1\circ{\mathring{\boldsymbol{Q}}}_2=\boldsymbol{1}+\epsilon \boldsymbol{0}$ and $\mathring{\boldsymbol{Q}}_1
=
\mathring{\boldsymbol{Q}}_2$.
Suppose that $\mathring{\boldsymbol{Q}}_1
=
\mathring{\boldsymbol{Q}}_2$ holds, we can directly obtain $\mathring{\boldsymbol{Q}}^*_1 \circ \mathring{\boldsymbol{Q}}_2=\boldsymbol{1}+\epsilon \boldsymbol{0}$ according to \eqref{eq13_f}. Conversely, suppose that $\mathring{\boldsymbol{Q}}^*_1 \circ \mathring{\boldsymbol{Q}}_2=\boldsymbol{1}+\epsilon \boldsymbol{0}$ holds. We aim to show that this statement implies $\mathring{\boldsymbol{Q}}_1
=
\mathring{\boldsymbol{Q}}_2$. To do so, we write ${\mathring{\boldsymbol{Q}}}_i$ for $i=1,2$ in component form as
\begin{align*}
{\mathring{\boldsymbol{Q}}}_i
\triangleq
\left[\begin{array}{c}
\varepsilon_{r_i}+\epsilon \varepsilon_{c_i}
\\
{\boldsymbol{q}}_{r_i}+\epsilon {\boldsymbol{q}}_{c_i}
\end{array}\right]
\end{align*}
\noindent according to \eqref{eq11_f}. Note that ${\mathring{\boldsymbol{Q}}}_i\in\mathbb{DS}^3,\,\forall i=1,2$ leads to
\begin{align}
\label{eq56_f}
\varepsilon^2_{r_i}+\boldsymbol{q}^\top_{r_i}\boldsymbol{q}_{r_i}=1 \text{ and }
\varepsilon_{r_i}\varepsilon_{c_i}+\boldsymbol{q}^\top_{r_i}\boldsymbol{q}_{c_i}=0
\end{align} 
for $i=1,2$. From the fact that ${\mathring{\boldsymbol{Q}}}^*_1\circ{\mathring{\boldsymbol{Q}}}_2=\boldsymbol{1}+\epsilon \boldsymbol{0}$, the real part gives $\varepsilon_{r_1}\varepsilon_{r_2}+\boldsymbol{q}^\top_{r_1}\boldsymbol{q}_{r_2}=1$ holds. Combining this result with the first equation of \eqref{eq56_f}, we arrive at
\begin{align*}
\left(\varepsilon_{r_1}-\varepsilon_{r_2}\right)^2+\|\boldsymbol{q}_{r_1}-\boldsymbol{q}_{r_2}\|_2^2=0,
\end{align*}
which means
\begin{align}
\label{eq57_f}
\varepsilon_{r_1}=\varepsilon_{r_2} \text{ and }
\boldsymbol{q}_{r_1}=\boldsymbol{q}_{r_2}.
\end{align} 
\noindent Additionally, considering that the dual part of ${\mathring{\boldsymbol{Q}}}^*_1\circ{\mathring{\boldsymbol{Q}}}_2$ is equal to $\boldsymbol{0}$, we deduce that $\varepsilon_{r_1}\boldsymbol{q}_{c_2}
+
\varepsilon_{c_1}\boldsymbol{q}_{r_2}
-
\varepsilon_{r_2}\boldsymbol{q}_{c_1}
-
\varepsilon_{c_2}\boldsymbol{q}_{r_1}
-
\boldsymbol{q}^\times_{c_1}\boldsymbol{q}_{r_2}
-
\boldsymbol{q}^\times_{r_1}\boldsymbol{q}_{c_2}
=
\boldsymbol{0}$ holds, which can be rewritten as
\begin{align}
\label{eq58_f}
\varepsilon_{r_1}
\left(\boldsymbol{q}_{c_2}
-
\boldsymbol{q}_{c_1}\right)
+
\left(\varepsilon_{c_1}
-
\varepsilon_{c_2}\right)\boldsymbol{q}_{r_1}
=
\boldsymbol{q}^\times_{r_1}
\left(\boldsymbol{q}_{c_2}
-
\boldsymbol{q}_{c_1}\right)
\end{align}
according to \eqref{eq57_f}. Therefore, $\boldsymbol{q}_{c_1}=
\boldsymbol{q}_{c_2}$ holds because the item on the right-hand side of \eqref{eq58_f} is perpendicular to either of the two items on the left-hand side of it. This fact, together with \eqref{eq56_f}, further shows $\varepsilon_{c_1}
=
\varepsilon_{c_2}$. Thus, ${\mathring{\boldsymbol{Q}}}^*_1\circ{\mathring{\boldsymbol{Q}}}_2=\boldsymbol{1}+\epsilon \boldsymbol{0}$ is equivalent to $\mathring{\boldsymbol{Q}}_1
=
\mathring{\boldsymbol{Q}}_2$.

Expanding upon the same line of argument as presented above, we can establish the equivalence between ${\mathring{\boldsymbol{Q}}}^*_1\circ{\mathring{\boldsymbol{Q}}}_2=-\boldsymbol{1}+\epsilon \boldsymbol{0}$ and $\mathring{\boldsymbol{Q}}_1
=
-\mathring{\boldsymbol{Q}}_2$. Thus, its detailed analysis is omitted. The proof of Lemma \ref{lem1} is complete.
\end{proof}

\section{Proof of Lemma \ref{lem2}}
\label{appendix:lem2}
\begin{proof}
Given a dual quaternion $\mathring{\boldsymbol{Q}}$ defined in Lemma \ref{lem2}. According to \eqref{eq15_f}, it is obtained that
\begin{align*}
\mathring{\boldsymbol{Q}}\circ\mathring{\boldsymbol{Q}}^{*}
\!=\!
\left(\boldsymbol{Q}+\epsilon \frac{1}{2} \boldsymbol{Q} \circ \operatorname{aug}\left(\boldsymbol{P}\right)
\right)
\!\circ\!
\left(
\boldsymbol{Q}+\epsilon \frac{1}{2} \boldsymbol{Q} \circ \operatorname{aug}\left(\boldsymbol{P}\right)
\right)^{*}\!\!.
\end{align*}
Considering \eqref{eq7_f}, \eqref{eq12_f}, and the definition of $\operatorname{aug}\left(\cdot\right)$ at the end of Section \ref{sec_Introduction}, we have
\begin{align*}
\left(
\boldsymbol{Q}+\epsilon \frac{1}{2} \boldsymbol{Q} \circ \operatorname{aug}\left(\boldsymbol{P}\right)
\right)^{*}
=
\boldsymbol{Q}^{*}
-
\epsilon \frac{1}{2}  \operatorname{aug}\left(\boldsymbol{P}\right)
\circ
\boldsymbol{Q}^{*},
\end{align*}
\noindent which, together with distributivity of the quaternion multiplication and $\epsilon^2=0$, leads to
\begin{align*}
\mathring{\boldsymbol{Q}}\circ\mathring{\boldsymbol{Q}}^{*}
=
\boldsymbol{Q}
\circ
\boldsymbol{Q}^{*}.
\end{align*}
Therefore, $\mathring{\boldsymbol{Q}}\circ\mathring{\boldsymbol{Q}}^{*}
=
\boldsymbol{1}+\epsilon \boldsymbol{0}$ holds due to $\mathring{\boldsymbol{Q}}\in\mathbb{S}^3$, which implies that $\mathring{\boldsymbol{Q}}$ defined in \eqref{eq15_f} is a unit dual quaternion.

Obviously, $\boldsymbol{Q}$ can be obtained as the real part. Thus, $\boldsymbol{Q}^*$ can be known easily. In addition, we have
\begin{align*}
2\boldsymbol{Q}^*\circ\left(\frac{1}{2} \boldsymbol{Q} \circ \operatorname{aug}\left(\boldsymbol{P}\right)\right)=\operatorname{aug}\left(\boldsymbol{P}\right)
\end{align*} 
because of associativity of the multiplication of quaternions. As a result, $\boldsymbol{Q}$ and $\boldsymbol{P}$ can be uniquely determined by $\mathring{\boldsymbol{Q}}$. The proof of Lemma \ref{lem2} is complete.
\end{proof}

\section{Proof of Lemma \ref{lem4}}
\label{appendix:lem4}
\begin{proof}
It is known that ${\mathring{\boldsymbol{Q}}}$ has the same form as \eqref{eq11_f}. Based on the conjugation definitions of $\mathring{\boldsymbol{Q}}$ in \eqref{eq12_f} and $\operatorname{aug}(\cdot)$, we obtain 
\begin{align}
\label{eq59_f}
\mathring{\boldsymbol{y}}
=
\left[\begin{array}{c}
\mathring{\boldsymbol{q}}^\top\mathring{\boldsymbol{x}}\mathring{\varepsilon}-\mathring{\boldsymbol{q}}^\top\mathring{\boldsymbol{x}}\mathring{\varepsilon}+\mathring{\boldsymbol{q}}^\top\mathring{\boldsymbol{q}}^\times\mathring{\boldsymbol{x}}
\\
\mathring{\boldsymbol{x}}\mathring{\varepsilon}\mathring{\varepsilon}
+
\mathring{\boldsymbol{q}}\mathring{\boldsymbol{q}}^\top\mathring{\boldsymbol{x}}
-2\mathring{\boldsymbol{q}}^\times\mathring{\boldsymbol{x}}\mathring{\varepsilon}
+\mathring{\boldsymbol{q}}^\times\mathring{\boldsymbol{q}}^\times\mathring{\boldsymbol{x}}
\end{array}\right].
\end{align}
\noindent Noting that $\mathring{\boldsymbol{q}}^\top\mathring{\boldsymbol{q}}^\times=\boldsymbol{0}+\epsilon \boldsymbol{0}$ and $\mathring{\boldsymbol{q}}^\times\mathring{\boldsymbol{q}}^\times\mathring{\boldsymbol{x}}
=
\mathring{\boldsymbol{q}}\mathring{\boldsymbol{q}}^\top\mathring{\boldsymbol{x}}-\mathring{\boldsymbol{x}}\mathring{\boldsymbol{q}}^\top\mathring{\boldsymbol{q}}$, we rewrite \eqref{eq59_f} as
\begin{align*}
\mathring{\boldsymbol{y}}
=
\left[\begin{array}{c}
0+\epsilon 0
\\
\mathring{\boldsymbol{x}}\left(\mathring{\varepsilon}\mathring{\varepsilon}-\mathring{\boldsymbol{q}}^\top\mathring{\boldsymbol{q}}\right)
+
2\mathring{\boldsymbol{q}}\mathring{\boldsymbol{q}}^\top\mathring{\boldsymbol{x}}
-2\mathring{\boldsymbol{q}}^\times\mathring{\boldsymbol{x}}\mathring{\varepsilon}
\end{array}\right],
\end{align*}
\noindent which shows
\begin{align*}
\boldsymbol{y}_r
=
\left[\begin{array}{c}
0
\\
\left(\varepsilon_r^2-\boldsymbol{q}_r^\top\boldsymbol{q}_r\right)\boldsymbol{x}_{r}+2\boldsymbol{q}_r\boldsymbol{q}^\top_r\boldsymbol{x}_{r}-2\varepsilon_r\boldsymbol{q}^\times_r\boldsymbol{x}_{r}
\end{array}\right].
\end{align*}
\noindent It is not difficult to verify that $\left(\varepsilon_r^2-\boldsymbol{q}_r^\top\boldsymbol{q}_r\right)\boldsymbol{I}_{3}+2\boldsymbol{q}_r\boldsymbol{q}^\top_r-2\varepsilon_r\boldsymbol{q}^\times_r$ is orthogonal. This fact, together with $\varepsilon_r^2+\boldsymbol{q}_r^\top\boldsymbol{q}_r=1$ from the definition of $\mathbb{DS}^3$, leads to $\|\boldsymbol{x}_r\|_2=\|\boldsymbol{y}_r\|_2$ naturally. The proof of Lemma \ref{lem4} is complete.
\end{proof}

\section{Derivation of \eqref{eq27_f}}
\label{appendix:eq25}

Based on \eqref{eq26_f}, we have
\begin{align*}
\delta\dot{\mathring{\boldsymbol{\omega}}}_k(t)
\triangleq{}& \dot{\mathring{\boldsymbol{\omega}}}_k(t)
-
\operatorname{tr}
\left(\dot{\mathring{\boldsymbol{\omega}}}_{d}(t)
\right)
\nonumber
\\
&
-
\operatorname{red}
\left(\delta\dot{\mathring{\boldsymbol{Q}}}^{*}_k(t)
\circ
\operatorname{aug}\left(\mathring{\boldsymbol{\omega}}_{d}(t)\right)
\circ
\delta\mathring{\boldsymbol{Q}}_k(t)\right)
\nonumber
\\
&
-
\operatorname{red}
\left(\delta\mathring{\boldsymbol{Q}}^{*}_k(t)
\circ
\operatorname{aug}\left(\mathring{\boldsymbol{\omega}}_{d}(t)\right)
\circ
\delta\dot{\mathring{\boldsymbol{Q}}}_k(t)\right).
\end{align*}
Additionally, it is obtained that
\begin{align*}
\delta\dot{\mathring{\boldsymbol{Q}}}^{*}_k(t)
=
-\frac{1}{2}
\operatorname{aug}\left(
\delta\mathring{\boldsymbol{\omega}}_k(t)\right)
\circ
\delta\mathring{\boldsymbol{Q}}^{*}_k(t),
\end{align*}
based on \eqref{eq14_f}, \eqref{eq25_f}, and the definition of $\operatorname{aug}(\cdot)$. Considering \eqref{eq24_f} and \eqref{eq26_f} again, we have
\begin{align*}
\delta\dot{\mathring{\boldsymbol{\omega}}}_k(t)
\triangleq{}& \dot{\mathring{\boldsymbol{\omega}}}_k(t)
+
\frac{1}{2}\operatorname{red}
\left(
\operatorname{aug}\left(
\delta\mathring{\boldsymbol{\omega}}_k(t)\right)
\circ
\operatorname{aug}\left(
\operatorname{tr}\left(\mathring{\boldsymbol{\omega}}_{d}(t)\right)
\right)
\right)
\nonumber
\\
&
-
\operatorname{tr}
\left(\dot{\mathring{\boldsymbol{\omega}}}_{d}(t)
\right)
-
\frac{1}{2}\operatorname{red}
\big(
\operatorname{aug}\left(
\operatorname{tr}\left(\mathring{\boldsymbol{\omega}}_{d}(t)\right)
\right)
\nonumber
\\
&\quad\quad\quad\quad\quad\quad\quad\quad\quad~~
\circ
\operatorname{aug}\left(
\delta\mathring{\boldsymbol{\omega}}_k(t)\right)
\big),
\end{align*}
\noindent where Lemma \ref{lem4} is applied. Furthermore, according to \eqref{eq13_f} and the definitions of $\operatorname{red}(\cdot)$ and $(\cdot)^\times$, we have
\begin{align*}
\delta\dot{\mathring{\boldsymbol{\omega}}}_k(t)
\triangleq{}& \dot{\mathring{\boldsymbol{\omega}}}_k(t)
+
\delta\mathring{\boldsymbol{\omega}}^\times_k(t)\operatorname{tr}\left(\mathring{\boldsymbol{\omega}}_{d}(t)\right)
-
\operatorname{tr}
\left(\dot{\mathring{\boldsymbol{\omega}}}_{d}(t)
\right)
.
\end{align*}
Observing the second equation of \eqref{eq26_f} again, we obtain
\begin{align*}
\delta\dot{\mathring{\boldsymbol{\omega}}}_k(t)
\triangleq{}& \dot{\mathring{\boldsymbol{\omega}}}_k(t)
-
\mathring{\boldsymbol{\omega}}^\times_{k}(t)
\delta\mathring{\boldsymbol{\omega}}_k(t)
-
\operatorname{tr}
\left(\dot{\mathring{\boldsymbol{\omega}}}_{d}(t)
\right),
\end{align*}
\noindent which indicates \eqref{eq27_f}. The derivation of \eqref{eq27_f} is complete.

\section{Proof of Lemma \ref{lem5}}
\label{appendix:lem5}
\begin{proof}
Suppose that $\mathring{\boldsymbol{x}}\triangleq \boldsymbol{x}_r+\epsilon \boldsymbol{x}_c$ and $\mathring{\boldsymbol{y}}\triangleq \boldsymbol{y}_r+\epsilon \boldsymbol{y}_c$, where $\boldsymbol{x}_r,\,\boldsymbol{x}_c,\,\boldsymbol{y}_r,\,\boldsymbol{y}_c\in\R^3$. Based on \eqref{eq9_f} and \eqref{eq22_f}, we have
\begin{align*}
\operatorname{crs}\left(\mathring{\boldsymbol{x}},-\mathring{\boldsymbol{M}}\left(\mathring{\boldsymbol{y}}^\times\mathring{\boldsymbol{x}}\right)\right)
=
-\boldsymbol{x}_r^\top\boldsymbol{J}\boldsymbol{y}_r^\times\boldsymbol{x}_r
-
m\boldsymbol{x}_c^\top\boldsymbol{y}_c^\times\boldsymbol{x}_r
\end{align*}
\noindent and
\begin{align*}
\operatorname{crs}\left(\mathring{\boldsymbol{x}},\mathring{\boldsymbol{y}}^\times\mathring{\boldsymbol{M}}\mathring{\boldsymbol{x}}\right)
&=
\boldsymbol{x}_r^\top\boldsymbol{y}_r^\times\boldsymbol{J}\boldsymbol{x}_r
+
m\boldsymbol{x}_r^\top\boldsymbol{y}_c^\times\boldsymbol{x}_c
\nonumber
\\
&=
-\boldsymbol{x}_r^\top\boldsymbol{J}\boldsymbol{y}_r^\times\boldsymbol{x}_r
-
m\boldsymbol{x}_c^\top\boldsymbol{y}_c^\times\boldsymbol{x}_r,
\end{align*}
\noindent where $\boldsymbol{x}_c^\top\boldsymbol{y}_r^\times\boldsymbol{x}_c=0$ and the symmetry of $\boldsymbol{J}$ are applied. Based on the above analysis, \eqref{eq29_f} can be obtained naturally. The proof of Lemma \ref{lem5} is complete.
\end{proof}

\section{Proof of Lemma \ref{lem6}}
\label{appendix:lem6}
\begin{proof}
According to \eqref{eq13_f}, it is obtained that $$\mathring{\boldsymbol{Q}}_1\circ \mathring{\boldsymbol{Q}}_2
=f \left(\mathring{\boldsymbol{Q}}_1\right) \mathring{\boldsymbol{Q}}_2,
$$ 
\noindent where
\begin{align*}
f (\mathring{\boldsymbol{Q}}_1)
\triangleq
\left[\begin{array}{cc}
\mathring{\varepsilon}_1 & -\mathring{\boldsymbol{q}}_1^\top
\\
\mathring{\boldsymbol{q}}_1 & \mathring{\varepsilon}_1 \boldsymbol{I}_{3} + \mathring{\boldsymbol{q}}_1^\times
\end{array}\right].
\end{align*}
\noindent It is not hard to verify
\begin{align*}
f^\top (\mathring{\boldsymbol{Q}}_1)=f (\mathring{\boldsymbol{Q}}^*_1).
\end{align*}
Based on these facts, we have
\begin{align*}
\left(\mathring{\boldsymbol{Q}}_1\circ \mathring{\boldsymbol{Q}}_2\right)^\top \mathring{\boldsymbol{Q}}_3
=
\left(
f \left(\mathring{\boldsymbol{Q}}_1\right) \mathring{\boldsymbol{Q}}_2
\right)^\top \mathring{\boldsymbol{Q}}_3
=
\mathring{\boldsymbol{Q}}_2^\top
f^\top \left(\mathring{\boldsymbol{Q}}_1\right)
\mathring{\boldsymbol{Q}}_3
,
\end{align*}
\noindent and thus,
\begin{align*}
\left(\mathring{\boldsymbol{Q}}_1\circ \mathring{\boldsymbol{Q}}_2\right)^\top \mathring{\boldsymbol{Q}}_3
=
\mathring{\boldsymbol{Q}}_2^\top
f(\mathring{\boldsymbol{Q}}^*_1)
\mathring{\boldsymbol{Q}}_3
=
\mathring{\boldsymbol{Q}}_2^\top
\left(\mathring{\boldsymbol{Q}}_1^*\circ \mathring{\boldsymbol{Q}}_3\right),
\end{align*}
\noindent which implies that the imaginary part of $\left(\mathring{\boldsymbol{Q}}_1\circ \mathring{\boldsymbol{Q}}_2\right)^\top \mathring{\boldsymbol{Q}}_3$ is equal to that of $
\mathring{\boldsymbol{Q}}_2^\top \left(\mathring{\boldsymbol{Q}}_1^*\circ \mathring{\boldsymbol{Q}}_3\right)$. This fact is consistent with the description in \eqref{eq30_f} according to the definition of $\operatorname{crs}\left(\cdot\right)$ in \eqref{eq9_f}. The proof of Lemma \ref{lem6} is complete.
\end{proof}

\section*{Acknowledgment}

The authors would like to thank Dr. Jingyao Zhang from Beihang University and Dr. Yu Hui from Qingdao University of Science and Technology for their helpful discussions.

\vfill

\end{document}